\documentclass[a4paper]{elsarticle} 
\usepackage[applemac]{inputenc}
\usepackage[english]{babel}
\usepackage{amssymb}
\usepackage{cmll}
\usepackage[T1]{fontenc}
\usepackage{amsmath}
\usepackage{amssymb,mathtools,stmaryrd,tikz,hyperref}
\usepackage{fouriernc}
\usepackage{enumitem}
\usepackage{subfig}

\usepackage{fullmacros}

\newcommand{\model}[2]{\mathbb{M}[#1,\mathfrak{#2}]}

\newcommand{\aeq}{=_{\textnormal{a.e.}}}
\newcommand{\preorder}[1][]{\mathcal{P}(#1)}

\newcommand{\words}[1]{\textnormal{words}(#1)}
\newcommand{\wordspos}[1]{\textnormal{words}^{+}(#1)}

\makeatletter
\providecommand{\doi}[1]{%
  \begingroup
    \let\bibinfo\@secondoftwo
    \urlstyle{rm}%
    \href{http://dx.doi.org/#1}{%
      doi:\discretionary{}{}{}%
      \nolinkurl{#1}%
    }%
  \endgroup
}
\makeatother

\begin{document}

\title{Towards a\\ \emph{Complexity-through-Realisability} Theory}
\author[ts]{Thomas Seiller\corref{cor1}\fnref{fn1}}
\ead{seiller@ihes.fr}
\cortext[cor1]{Corresponding author}
\fntext[fn1]{This work was partially supported by the the Marie Sk\l{}odowska-Curie Individual Fellowship (H2020-MSCA-IF-2014) 659920 - ReACT, as well as the French ANR project ANR-10-BLAN-0213 LOGOI and the french ANR project COQUAS (number 12 JS02 006 01).}
\address[ts]{Proofs, Programs, Systems. CNRS -- Paris Diderot University.}

\begin{abstract}
We explain how recent developments in the fields of Realisability models for linear logic \cite{seiller-goig} -- or \emph{geometry of interaction} -- and implicit computational complexity \cite{seiller-conl,seiller-lsp} can lead to a new approach of implicit computational complexity. This semantic-based approach should apply uniformly to various computational paradigms, and enable the use of new mathematical methods and tools to attack problem in computational complexity. This paper provides the background, motivations and perspectives of this \emph{complexity-through-Realisability} theory to be developed, and illustrates it with recent results \cite{seiller-goic}.
\end{abstract}

\maketitle

\section{Introduction}

Complexity theory lies at the intersection between mathematics and computer science, and studies the amount of resources needed to run a specific program (complexity of an algorithm) or solve a particular problem (complexity of a problem). 
wewill explain how it is possible to build on recent work in \emph{Realisability models for linear logic} -- a mathematical model of programs and their execution -- to provide new characterisations of existing complexity classes. It is hoped that these characterisations will enable new mathematical techniques, tools and invariants from the fields of operators algebras and dynamical systems, providing researchers with new tools and methods to attack long-standing open problems in complexity theory.

The \emph{complexity-through-Realisability} theory we propose to develop will provide a unified framework for studying many computational paradigms and their associated computational complexity theory grounded on well-studied mathematical concepts. This 
should 
provide a good candidate for a theory of complexity for computational paradigms currently lacking an established theory (e.g. concurrent processes), as well as contribute to establish a unified and well-grounded account of complexity for higher-order functionals.


 Even though it has been an established discipline for more than 50 years \cite{hartmanisstearns}, many questions in complexity theory, even basic ones, remain open. During the last twenty years, researchers have developed new approaches based on logic: they offer solid, machine-independent, foundations and provide new tools and methods. Amongst these approaches, the fields of Descriptive Complexity (DC) 
and Implicit Computational Complexity (ICC) lead to a number of new characterisations of complexity classes
. These works laid grounds for both theoretical results  \cite{immerman} and applications such as typing systems for complexity constrained programs and type inference algorithms for statically determining complexity bounds. 

The \emph{complexity-through-Realisability} theory we propose to develop is related to those established logic-based approaches. As such, it inherits their strengths: it is machine-independent, provides tools and methods from logic and gives grounds for the above mentioned applications. 
Furthermore, 
it builds on state-of-the-art theoretical results on Realisability models for linear logic \cite{seiller-goig} using well-studied mathematical concepts from operators algebras and dynamical systems. As a consequence, it opens the way to use against complexity theory's open problems the many techniques, tools and invariants that were developed in these disciplines. 


We illustrate the approach by explaining how first results were recently obtained by capturing a large family of complexity classes corresponding to various notions of automata. Indeed, we provided \cite{seiller-goic} Realisability models in which types of binary predicates correspond to the classes of languages accepted by one-way (resp. two-way) deterministic (resp. non-deterministic, resp. probabilistic) multi-head automata. This large family of languages contains in particular the classes \Regular (regular languages), \Stochastic (stochastic languages), \Logspace (logarithmic space), \NLogspace (non-deterministic logarithmic space), \coNLogspace (complementaries of languages in \NLogspace), and \PLogspace (probabilistic logarithmic space).

Finally, we discuss the possible extensions and further characterisations that could be obtained by the same methods. We in particular discuss the relationship between the classification of complexity classes and the problem of classifying a certain kind of algebras, namely \emph{graphing algebras}.






\subsection{Complexity Theory}



Complexity theory is concerned with the study of how many resources are needed to perform a specific computation or to solve a given problem. The study of \emph{complexity classes} -- sets of problems which need a comparable amount of resources to be solved, lies at the intersection of mathematics and computer science. Although a very active and established field for more than fifty years \cite{cobham,cookfoundations,hartmanisstearnsfoundations,savitch}, a number of basic problems remain open,  
for instance the famous \enquote{millennium problem} of whether \Ptime equals \NPtime or the less publicized but equally important question of whether \Ptime equals \Logspace. 
In recent years, several results have greatly modified the landscape of complexity theory by showing that proofs of separation (i.e. inequality) of complexity classes are hard to come by, pointing out the need to develop new theoretical methods. The most celebrated result in this direction \cite{naturalproofs} defines a notion of \emph{natural proof} comprising all previously developed proof methods and shows that no \enquote{natural proof} can succeed in proving separation. 

Mathematicians have then tried to give characterisations of complexity classes that differ from the original machine-bound definitions, hoping to enable methods from radically different areas of mathematics. 
Efforts in this direction lead to the development of Descriptive Complexity (DC), a field which studies the types of logics whose individual sentences characterise exactly particular complexity classes. Early developments were the 1974 Fagin-Jones-Selman results \cite{fagin74,jonesselman} characterising the classes \NExptime and \NPtime.
Many such characterisations have then been given \cite{comptongradel,dawargradel,gradelgurevich} and the method led Immerman to a proof of the celebrated Immerman-Szelepcs\'{e}nyi theorem \cite{immerman,szelepcsenyi} stating the two complexity classes \coNLogspace and \NLogspace are equal (though Szelepcs\'{e}nyi's proof does not use logic-based methods).

Implicit Computational Complexity (ICC) develops a related approach whose aim is to study algorithmic complexity only in terms of restrictions of languages and computational principles.
It has been established since Bellantoni and Cook' landmark paper \cite{bellantonicook}, and following work by Leivant and Marion \cite{leivantmarion1,leivantmarion2}.
Amongst the different approaches to ICC, several results were obtained by considering syntactic restrictions of \emph{linear logic} \cite{ll}, a refinement of intuitionnistic logic which accounts for the notion of resources. Linear logic introduces a modality $\oc$ marking the \enquote{possibility of duplicating} a formula $A$: the formula $A$ shall be used exactly once, while the formula $\oc A$ can be used any number of times. Modifying the rules governing this modality then yields variants of linear logic having computational interest: this is how constrained linear logic systems, for instance \BLL \cite{BLL} and \ELL \cite{danosjoinet}, are obtained. However, only a limited number of complexity classes were characterised in this way, and the method seems to be limited by its syntactic aspect: while it is easy to modify existing rules, it is much harder to find new, alternative, rules from scratch. The approach we propose to follow in this paper does not suffer from these limitations
, allowing for subtle distinctions unavailable to the syntactic techniques of ICC. 

\subsection{Realisability Models for Linear Logic} 

Concurrently to these developments in computational complexity, and motivated by disjoint questions and interests, Girard initiated the Geometry of Interaction (GoI) program \cite{towards}. This research program aims at obtaining particular kinds of \emph{Realisability} models (called GoI models) for linear logic. 
Realisability was first introduced \cite{kleene} as a way of making the Brouwer-Heyting-Kolmogorov interpretation of constructivism and intuitionistic mathematics precise; the techniques were then extended to classical logic, for instance by Krivine \cite{krivine}, and linear logic.
The GoI program quickly arose as a natural and well-suited tool for the study of computational complexity. Using the first GoI model \cite{goi1}, Abadi, Gonthier and L\'{e}vy \cite{AbadiGonthierLevy92b} showed the optimality of Lamping's reduction in lambda-calculus \cite{Lamping90}. It was also applied in implicit computational complexity \cite{baillotpedicini}, and was the main inspiration behind dal Lago's context semantics \cite{Lago}. 


More recently the geometry of interaction program inspired new techniques in implicit computational complexity. These new methods were initiated by Girard \cite{normativity} and have known a rapid development. They lead to a series of results in the form of new characterisations of the classes \coNLogspace \cite{normativity,seiller-conl}, \Logspace \cite{seiller-lsp,aplas14} and \Ptime \cite{lics-ptime}. 
Unfortunately, although the construction of Realisability models and the characterisations of classes are founded on similar techniques, they are two distinct, unrelated, constructions. 
The approach we propose to develop here will in particular bridge this gap and provide similar characterisations which will moreover allow the use of both logical and Realisability-specific methods.


\section{A \emph{Complexity-through-Realisability} Theory}

\subsection{Technical Background and Motivations}

About ten years ago, Girard showed \cite{feedback} that the restriction to the unit ball of a von Neumann algebra of the so-called \enquote{feedback equation}, which represents the execution of programs in GoI models, always has a solution. 
Moreover, previous and subsequent work showed the obtained GoI model interprets, depending on the choice of the von Neumann algebra, either full linear logic \cite{goi3} or the constrained system ELL which characterises elementary time computable functions \cite{seiller-phd}.
This naturally leads to the informal conjecture that there should be a correspondence between von Neumann algebras and complexity constraints. 

This deep and promising idea turned out to be slightly inexact and seemingly difficult to exploit. Indeed, the author showed  \cite{seiller-phd,seiller-masas} that the expressivity of the logic interpreted in a GoI model depends not only on the enveloping von Neumann algebra $\vn{N}$ but also on a maximal abelian sub-algebra (masa) $\vn{A}$ of $\vn{N}$, 
hinting at a refined conjecture stating that complexity constraints correspond to such couples $(\vn{A},\vn{N})$.
This approach is however difficult to extend 
to exploit and adapt to model other constrained logical systems for two reasons. The first reason is that the theory of maximal abelian sub-algebras in von Neumann algebras is an involved subject matter
 still containing large numbers of basic but difficult open problems \cite{FiniteVNAandMasas}. The second is that even though some results were obtained, no intuitions were gained about what makes the correspondence between couples $(\vn{A},\vn{N})$ and complexity constraints work.
 
Some very recent work of the author provides the foundational grounds for a new, tractable way of exploring the latter refined conjecture.
This series of work \cite{seiller-goim,seiller-goia,seiller-goig,seiller-goie} describes a systematic construction of realisability models for linear logic which unifies and extends all GoI models introduced by Girard. 
The construction is built upon a generalization of graphs, named \emph{graphings} \cite{adams,levitt_graphings,gaboriaucost}, which can be understood either as \emph{geometric realisations} of graphs on a measure space $(X,\mathcal{B},\mu)$, as measurable families of graphs, or as generalized measured dynamical system. It is parametrized by two monoids describing the model of computation and a map describing the realisability structure: 
\begin{itemize}[nolistsep,noitemsep]
\item a monoid $\Omega$ used to associate weights to edges of the graphs;
\item a map $m:\Omega\rightarrow\bar{\mathbf{R}}_{\geqslant 0}$ defining \emph{orthogonality} -- accounting for linear negation;
\item a monoid $\microcosm{m}$ -- the \emph{microcosm} -- of measurable maps from $(X,\mathcal{B},\mu)$ to itself.
\end{itemize}
A \emph{$\Omega$-weighted graphing in $\microcosm{m}$} is then defined as a directed graph $F$ whose edges are weighted by elements in $\Omega$, whose vertices are measurable subsets of the measurable space $(X,\mathcal{B})$, and whose edges are \emph{realised} by elements of $\microcosm{m}$, i.e. for each edge $e$ there exists an element $\phi_{e}$ in $\microcosm{m}$ such that $\phi_{e}(s(e))=t(e)$, where $s,t$ denote the source and target maps. Based on this notion, and an orthogonality relation defined from the map $m$, we obtained a systematic method for constructing realisability models for linear logic.

\begin{theorem}[Seiller \cite{seiller-goig}]\label{mainthm}
For all choices of $\Omega$, $\microcosm{m}$ and $m:\Omega\rightarrow\bar{\mathbf{R}}_{\geqslant 0}$, the set of $\Omega$-weighted graphings in $\microcosm{m}$ defines a model of Multiplicative-Additive Linear Logic (\MALL).
\end{theorem}

Let us notice that a microcosm $\microcosm{m}$ which contains only measure-preserving maps generates a measurable equivalence relation which, by the Feldman-Moore construction \cite{FeldmanMoore1}, induces a couple $(\vn{A},\vn{N})$ where $\vn{A}$ is a maximal abelian subalgebra of the von Neumann algebra $\vn{N}$. The previous theorem thus greatly generalizes Girard's general solution to the feedback equation \cite{feedback}, and in several ways. First, we define several models (deterministic, probabilistic, non-deterministic), among which some are unavailable to Girard's approach: for instance the non-deterministic model violates Girard's norm condition, and more generally Girard's techniques only apply when $\Omega$ is chosen as a subset of the complex plane unit disc.
It shows in fact much more as it exhibits an infinite family of structures of realisability models (parametrized by the map $m:\Omega\rightarrow\realposN\cup\{\infty\}$) on any model obtained from a microcosm. This extends drastically Girard's approach for which only two such structures were defined until now: an \enquote{orthogonality-as-nilpotency} structure in the algebra $\B{\hil{H}}$ \cite{goi1} and another one defined using the Fuglede-Kadison determinant \cite{FKdet} in the type {II}$_{1}$ hyperfinite factor \cite{goi5}.

\subsection{Methodology}

We can now explain the proposed methodology for defining a \emph{Complexity-through-realisability} theory.

The notion of $\Omega$-weighted $\microcosm{m}$-graphings for given monoid $\Omega$ and microcosm $\microcosm{m}$ yields a very general yet tractable mathematical notion of algorithm, as the microcosm $\microcosm{m}$ can naturally be understood as a set of computational principles \cite{seiller-lcc14}. It therefore provides an interesting middle-ground between usual computational models, for instance automata, and mathematical techniques from operator algebras and dynamical systems. It comprises Danos' interpretation of pure lambda-calculus (a Turing complete model of computation) in terms of operators \cite{danos-phd}, but it is not restricted to sequential algorithms as it will be shown to provide an interpretation of probabilistic and quantum programs. It also provides characterisations of usual complexity classes as types of predicates over binary words $\ListType\Rightarrow \cond{Bool}$, which will lead to a partial proof of the above conjecture by showing a correspondence between families of microcosms and complexity constraints. 


Work in this direction will establish these definitions of algorithms and complexity constraints as a uniform, homogeneous, machine-independent approach to complexity theory. The methods developed in this setting, either adapted from DC/ICC or specific to the realisability techniques employed, will apply to probabilistic/quantum complexity classes as much as sequential classes. In particular, it will offer a framework where comparison between non-classical and classical classes can be performed. It will also expand to computational paradigms where no established theory of complexity exists, providing a strong and coherent proposition for such.

It will extend the approach of ICC and DC as it will go beyond the syntactical restrictions they are suffering from. In particular, it will provide a new method for defining logical systems corresponding to complexity classes: the realisability model construction gives a systematic way to define a logic corresponding to the underlying computational model. It will also extend the GoI model approach to complexity by reconciling the logical and complexity aspects, allowing the use of both logical and realisability-specific methods.

Lastly, the approach we propose to develop does not naturally fall into the usual pitfalls for the obtention of separation results. Therefore, it provides a framework which will potentially offer separation methods, e.g. using invariants for the well-established mathematical notions it is founded upon.

\section{Interaction Graphs Models of Linear Logic}

\subsection{Graphings}


\begin{definition}
Let $(X,\mathcal{B},\lambda)$ be a measure space. We denote by $\mathcal{M}(X)$ the set of non-singular transformations\footnote{A non-singular transformation $f:X\rightarrow X$ is a measurable map which preserves the sets of null measure, i.e. $\lambda(f(A))=0$ if and only if $\lambda(A)=0$.} $X\rightarrow X$. A \emph{microcosm}  of the measure space $X$ is a subset $\microcosm{m}$ of $\mathcal{M}(X)$ which is closed under composition and contains the identity.
\end{definition}

In the following, we will consider a notion of graphing depending on a \emph{weight-monoid} $\Omega$, i.e. a monoid $(\Omega,\cdot,1)$ which contains the possible weights of the edges. 

\begin{definition}[Graphings]
Let $\microcosm{m}$ be a microcosm of a measure space $(X,\mathcal{B},\lambda)$ and $V^{F}$ a measurable subset of $X$. A \emph{$\Omega$-weighted graphing in $\microcosm{m}$} of carrier $V^{F}$ is a countable family $F=\{(\omega_{e}^{F},\phi_{e}^{F}: S_{e}^{F}\rightarrow T_{e}^{F})\}_{e\in E^{F}}$, where, for all $e\in E^{F}$ (the set of \emph{edges}):
\begin{itemize}[noitemsep,nolistsep]
\item $\omega_{e}^{F}$ is an element of $\Omega$, the \emph{weight} of the edge $e$;
\item $S_{e}^{F}\subset V^{F}$ is a measurable set, the \emph{source} of the edge $e$;
\item $T_{e}^{F}=\phi_{e}^{F}(S_{e}^{F})\subset V^{F}$ is a measurable set, the \emph{target} of the edge $e$;
\item $\phi_{e}^{F}$ is the restriction of an element of $\microcosm{m}$ to $S_{e}^{F}$, the \emph{realisation} of the edge $e$.
\end{itemize}
\end{definition}

It was shown in earlier work \cite{seiller-goia} how one can construct models of \MALL where proofs are interpreted as graphs. This construction relied on a single property, called the \emph{trefoil property}, which relates two simple notions: 
\begin{itemize}[noitemsep,nolistsep]
\item the \emph{execution} $F\plug G$ of two graphs, a graph defined as a set of paths;
\item the \emph{measurement} $\meas{F,G}$, a real number computed from a set of cycles.
\end{itemize}
These constructions can be extended to the more general framework where proofs are interpreted as graphings. Indeed, the notions of paths and cycles in a graphings are quite natural, and from two graphings $F,G$ in a microcosm $\microcosm{m}$ one can define its execution $F\plug G$ which is again a graphing in $\microcosm{m}$\footnote{As a consequence, a microcosm characterises a \enquote{closed world} for the execution of programs.}. A more involved argument then shows that the trefoil property holds for a family of measurements $\meas{\cdot,\cdot}$, where $m:\Omega\rightarrow\realposN\cup\{\infty\}$ is any measurable map. These results are obtained as a generalization of constructions considered in the author's PhD thesis\footnote{In the cited work, the results were stated in the particular case of the microcosm of measure-preserving maps on the real line.} \cite{seiller-phd}.

\begin{theorem}[Enveloping model]
Let $\Omega$ be a monoid and $\microcosm{m}$ a microcosm. The set of $\Omega$-weighted graphings in $\microcosm{m}$ yields a model, denoted by $\model{\Omega}{m}$, of \MALL.
\end{theorem}

In most of the models, one can define some exponential connectives. In particular, all models considered later on have the necessary structure to define an exponential modality. Let us notice however that the notion of exponential modality we are considering here is extremely weak, as most models won't validate the functorial promotion rule. The only rule that is assured to be satisfied by the exponential connectives we will consider is the contraction rule, i.e. for any type $\cond{A}$, one has $\oc\cond{A}\multimap \oc\cond{A}\otimes\oc\cond{A}$. These very weak exponential connectives will turn out to be of great interest: we obtain in this way models of linear logic where the exponentials are weaker than what is obtained through syntactic consideration in systems like \BLL, \SLL, etc. and characterise low complexity classes.

\subsection{Models of Computation}

Before explaining how one can characterise complexity classes in this way, we need to state refinements of the previous theorem. We first define the notion of \emph{deterministic graphing}.

\begin{definition}[Deterministic graphings]
A $\Omega$-weighted graphing $G$ is \emph{deterministic} when:
\begin{itemize}[nolistsep,noitemsep]
\item for all $e\in E^{G}$, $\omega^{G}_{e}\in\{0,1\}$;
\item the following set is of null measure: $\{x\in\measured{X}~|~\exists e\neq e'\in E^{G}, x\in S_{e}^{G}\cap S^{G}_{e'}\}$
\end{itemize}
\emph{Non-deterministic graphings} are defined as those graphings satisfying the first condition.
\end{definition}

We then prove that the notions of deterministic and non-deterministic graphings are closed under composition, i.e. if $F,G$ are deterministic graphings, then their execution $F\plug G$ is again a deterministic graphing. This shows that the sets of deterministic and non-deterministic graphings define submodels of $\model{\Omega}{m}$.

\begin{theorem}[Deterministic model]
Let $\Omega$ be a monoid and $\microcosm{m}$ a microcosm. The set of $\Omega$-weighted \emph{deterministic} graphings in $\microcosm{m}$ yields a model, denoted by $\dmodel{\Omega}{m}$, of multiplicative-additive linear logic.
\end{theorem}

\begin{theorem}[Non-deterministic model]
The set of $\Omega$-weighted \emph{non-deterministic} graphings in $\microcosm{m}$ yields a model, denoted by $\nmodel{\Omega}{m}$, of multiplicative-additive linear logic.
\end{theorem}

One can also consider several other classes of graphings. We explain here the simplest non-classical model one could consider, namely that of \emph{probabilistic graphings}. In order for this notion to be of real interest, one should suppose that the unit interval $[0,1]$ endowed with multiplication is a submonoid of $\Omega$.

\begin{definition}[Probabilistic graphings]
A $\Omega$-weighted graphing $G$ is \emph{probabilistic} when:
\begin{itemize}[nolistsep,noitemsep]
\item for all $e\in E^{G}$, $\omega^{G}_{e}\in[0,1]$;
\item the following set is of null measure: $\{x\in\measured{X}~|~\sum_{e\in E^{G},~x\in S^{G}_{e}} \omega^{G}_{e}> 1\}$
\end{itemize}
\end{definition}

It turns out that this notion of graphing also behaves well under composition, i.e. there exists a \emph{probabilistic} submodel of $\model{\Omega}{m}$, namely the model of \emph{probabilistic graphings}.

\begin{theorem}[Probabilistic model]
Let $\Omega$ be a monoid and $\microcosm{m}$ a microcosm. The set of $\Omega$-weighted \emph{probabilistic} graphings in $\microcosm{m}$ yields a model, denoted by $\pmodel{\Omega}{m}$, of multiplicative-additive linear logic.
\end{theorem}

%
%
%

\begin{figure}
\subfloat[Different types of graphings]{
\begin{tikzpicture}
	\node (model) at (0,0) {$\model{\Omega}{m}$};
	\node (dmodel) at (0,-2) {$\dmodel{\Omega}{m}$};
	\node (nmodel) at (-1,-1) {$\nmodel{\Omega}{m}$};
	\node (pmodel) at (1,-1) {$\pmodel{\Omega}{m}$};
	\node (a) at (2.1,-1) {~};
	
	\draw[>->] (dmodel) -- (nmodel) {};
	\draw[>->] (dmodel) -- (pmodel) {};
	\draw[>->] (nmodel) -- (model) {};
	\draw[>->] (pmodel) -- (model) {};
\end{tikzpicture}
}
\subfloat[Different microcosms]{
\begin{tikzpicture}
	\node (model) at (0,0) {$\model{\Omega}{m+n}$};
	\node (dmodel) at (0,-2) {$\model{\Omega}{m\cap n}$};
	\node (nmodel) at (-1,-1) {$\model{\Omega}{m}$};
	\node (pmodel) at (1,-1) {$\model{\Omega}{n}$};
	
	\draw[>->] (dmodel) -- (nmodel) {};
	\draw[>->] (dmodel) -- (pmodel) {};
	\draw[>->] (nmodel) -- (model) {};
	\draw[>->] (pmodel) -- (model) {};
\end{tikzpicture}
}
\subfloat[Different weight monoids]{
\begin{tikzpicture}
	\node (a) at (-2,-1) {~};
	\node (model) at (0,0) {$\model{\Theta+_{\Omega}\Xi}{m}$};
	\node (dmodel) at (0,-2) {$\model{\Omega}{m}$};
	\node (nmodel) at (-1,-1) {$\model{\Theta}{m}$};
	\node (pmodel) at (1,-1) {$\model{\Xi}{m}$};
	
	\draw[>->] (dmodel) -- (nmodel) {};
	\draw[>->] (dmodel) -- (pmodel) {};
	\draw[>->] (nmodel) -- (model) {};
	\draw[>->] (pmodel) -- (model) {};
\end{tikzpicture}
}
%
\caption{Inclusions of models}\label{inclusions}
\end{figure}
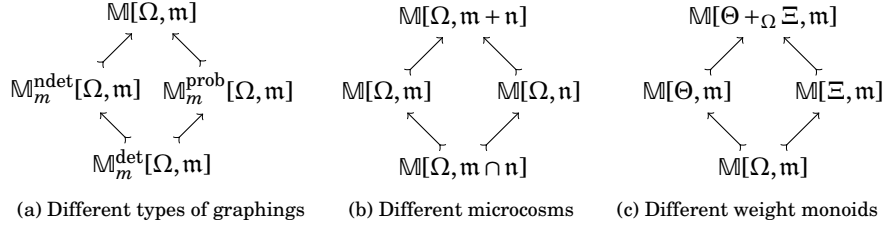

These models are all submodels of the single model $\model{\Omega}{m}$. Moreover, other inclusions of models can be obtained by modifying the other parameters, namely the weight monoid $\Omega$ and the microcosm $\microcosm{m}$.
For instance, given two microcosms $\microcosm{m}\subset\microcosm{n}$, it is clear that a graphing in $\microcosm{m}$ is in particular a graphing in $\microcosm{n}$. This inclusion actually extends to an embedding of the model $\model{\Omega}{m}$ into $\model{\Omega}{n}$ which preserves most logical operations\footnote{It preserves all connectives except for negation.}. Moreover, given two microcosms $\microcosm{m}$ and $\microcosm{n}$, one can define the \emph{smallest common extension} $\microcosm{m+n}$ as the compositional closure of the set $\microcosm{m}\cup\microcosm{n}$. The model $\model{\Omega}{m+n}$ then contains both models $\model{\Omega}{m}$ and $\model{\Omega}{n}$ through the embedding just mentioned.
In the same way, an inclusion of monoids $\Omega\subset\Gamma$ yields an embedding of the the model $\model{\Omega}{m}$ into $\model{\Gamma}{m}$. For instance, the model $\model{\{1\}}{m}$ is a submodel of $\model{\Omega}{m}$ for any monoid $\Omega$. One can also define, given weight monoids $\Omega$, $\Theta$ and $\Xi$ with monomorphisms $\Omega\rightarrow\Theta$ and $\Omega\rightarrow\Xi$, the model $\model{\Theta+_{\Omega}\Xi}{m}$ where $\Theta+_{\Omega}\Xi$ denotes the amalgamated sum of the monoids. \autoref{inclusions} illustrates some of these inclusions of models.

\section{characterisations of sub-logarithmic classes}\label{recentresults}

We now expose some recent results obtained by applying the methodology described above \cite{seiller-goic}. We describe in this way a number of new characterisations of sublogarithmic complexity classes. Before going into details about these characterisations, let us define a number of complexity classes -- all of them definable by classes of automata.

\begin{definition}
For each integer $i$, we define:
\begin{itemize}[nolistsep,noitemsep]
\item the class \cctwdfa{i} (resp. \ccowdfa{i}) as the set of languages accepted by deterministic two-way (resp. one-way) multihead automata with at most $i$ heads;
\item the class \cctwnfa{i} (resp. \ccownfa{i}) as the set of languages accepted by two-way (resp. one-way) multihead automata with at most $i$ heads;
\item the class \cctwconfa{i} (resp. \ccowconfa{i}) as the set of languages whose complementary language is accepted by two-way (resp. one-way) multihead automata with at most $i$ heads;
\item the class \cctwpfa{i} (resp. \ccowpfa{i}) as the set of languages accepted by two-way (resp. one-way) probabilistic multihead automata with at most $i$ heads;
\end{itemize}
We also denote by \Logspace (resp. \Ptime) the class of predicates over binary words that are recognized by a Turing machine using logarithmic space (resp. polynomial time), by \NLogspace (resp. \NPtime) its non-deterministic analogue, by \coNLogspace (resp. \coNPtime) the set of languages whose complementary language lies in \Logspace (resp. \Ptime). We also denote by \PLogspace the class of predicates over binary words that are recognized by a probabilistic Turing machine with unbounded error using logarithmic space.
\end{definition}

We don't recall the usual definitions of these variants of multihead automata, which can be easily found in the literature. We only recall the classical results:
\[
\begin{array}{ccc}
\cup_{i\in\naturalN}\text{\cctwdfa{i}}=\text{\Logspace}&\hspace{1cm}&\cup_{i\in\naturalN}\text{\cctwnfa{i}}=\text{\NLogspace}\\
\cup_{i\in\naturalN}\text{\cctwconfa{i}}=\text{\coNLogspace}&\hspace{1cm}&\cup_{i\in\naturalN}\text{\cctwpfa{i}}=\text{\PLogspace}
\end{array}
\]

In the following, we will denote by \cctwdfa{$\infty$} (resp. \cctwnfa{$\infty$}, resp. \cctwpfa{$\infty$}) the set $\cup_{i\in\naturalN}\text{\cctwdfa{i}}$ (resp. $\cup_{i\in\naturalN}\text{\cctwnfa{i}}$, resp. $\cup_{i\in\naturalN}\text{\cctwpfa{i}}$).

\subsection{Deterministic Computation: From Regular Languages to Logarithmic Space}

In the models of linear logic we described, one can easily define the type $\ListType$ of words over an arbitrary finite alphabet $\Sigma$. The definition of the representation of these binary words comes from the encoding of binary lists in lambda-calculus and is explained thoroughly in previous works \cite{seiller-phd,seiller-conl,seiller-lsp}. We won't give the formal definition of what is a representation of a word $\word{w}$ here, but let us sketch the main ideas. Given a word, say the binary word $\word{w}=00101$, we introduce a symbol $\star$ that can be understood as a left-hand end-of-tape marker and consider the list of symbols $\star 00101$. Then, the graphing that will represent $\word{w}$ is obtained as a realisation of the directed graph whose set of vertices is $\{\star,0,1\}\times\{\In,\Out\}$ and whose edges link the symbols of the list together, i.e. the graph pictured in \autoref{twintegerrep}.

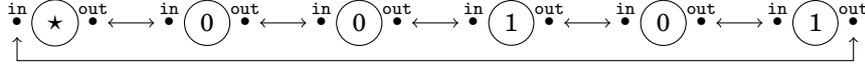
\begin{figure}
\begin{tikzpicture}
	\node[draw,circle] (0) at (0,0) {$\star$};
		\node (0i) at (-0.5,0) {$\bullet$};
			\node (0ilabel) at (0i.north) {\scriptsize{$\In$}};
		\node (0o) at (0.5,0) {$\bullet$};
			\node (0olabel) at (0o.north) {\scriptsize{$\Out$}};
	\node[draw,circle] (1) at (2,0) {$0$};
		\node (1i) at (1.5,0) {$\bullet$};
			\node (1ilabel) at (1i.north) {\scriptsize{$\In$}};
		\node (1o) at (2.5,0) {$\bullet$};
			\node (1olabel) at (1o.north) {\scriptsize{$\Out$}};
	\node[draw,circle] (2) at (4,0) {$0$};
		\node (2i) at (3.5,0) {$\bullet$};
			\node (2ilabel) at (2i.north) {\scriptsize{$\In$}};
		\node (2o) at (4.5,0) {$\bullet$};
			\node (2olabel) at (2o.north) {\scriptsize{$\Out$}};
	\node[draw,circle] (3) at (6,0) {$1$};
		\node (3i) at (5.5,0) {$\bullet$};
			\node (3ilabel) at (3i.north) {\scriptsize{$\In$}};
		\node (3o) at (6.5,0) {$\bullet$};
			\node (3olabel) at (3o.north) {\scriptsize{$\Out$}};
	\node[draw,circle] (4) at (8,0) {$0$};
		\node (4i) at (7.5,0) {$\bullet$};
			\node (4ilabel) at (4i.north) {\scriptsize{$\In$}};
		\node (4o) at (8.5,0) {$\bullet$};
			\node (4olabel) at (4o.north) {\scriptsize{$\Out$}};
	\node[draw,circle] (5) at (10,0) {$1$};
		\node (5i) at (9.5,0) {$\bullet$};
			\node (5ilabel) at (5i.north) {\scriptsize{$\In$}};
		\node (5o) at (10.5,0) {$\bullet$};
			\node (5olabel) at (5o.north) {\scriptsize{$\Out$}};
	
	\draw[<->] (0o) -- (1i) {};
	\draw[<->] (1o) -- (2i) {};
	\draw[<->] (2o) -- (3i) {};
	\draw[<->] (3o) -- (4i) {};
	\draw[<->] (4o) -- (5i) {};
	\draw[<->] (5o) -- (10.5,-0.5) -- (-0.5,-0.5) -- (0i) {};
\end{tikzpicture}
\caption{Representation of the word $\word{w}=00101$}\label{twintegerrep}
\end{figure}

We are now interested in the elements of the type $\oc\ListTypeBin$. For each word $\word{w}$, there exists an element $\oc L_{\word{w}}$ in the type $\oc\ListTypeBin$ which represents it. We say that a graphing -- or \emph{program} -- $P$ of type $\oc\ListTypeBin\multimap \cond{Bool}$ \emph{accepts} the word $\word{w}$ when the execution $P\plug W_{\word{w}}$ is equal to the distinguished element $\tt true\rm\in\cond{Bool}$. The \emph{language} accepted by such a program $P$ is then defined as $[P]=\{\word{w}\in\ListTypeBin~|~\phi\plug W_{\word{w}}=\tt true\rm\}$.

\begin{definition}[characterisation - deterministic models]
Let $\Omega$ be a monoid, $\microcosm{m}$ a microcosm and $\mathcal{L}$ a set of languages. We say the model $\genmodel[]{m}{\Omega}{\textnormal{det}}$ \emph{characterises the set $\mathcal{L}$} if the set $\{[P]~|~P\in \oc\ListTypeBin\multimap\cond{Bool}\}$ is equal to $\mathcal{L}$.
\end{definition}

We now consider the measure space $\integerN\times[0,1]^{\naturalN}$ endowed with the product of the counting measure on $\integerN$ and the Lebesgue measure on the Hilbert cube $[0,1]^{\naturalN}$. 
To define microcosms, we use the constructor $+$: if  $\microcosm{m}$ and $\microcosm{n}$ are two microcosms, $\microcosm{m+n}$ is the smallest microcosm containing both $\microcosm{m}$ and $\microcosm{n}$. We can now define the following microcosms:
\begin{itemize}[noitemsep,nolistsep]
\item $\microcosm{m}_{1}$ is the monoid of translations $\tau_{k}:(n,x)\mapsto (n+k,x)$;
\item $\microcosm{m}_{i+1}$ is the monoid $\microcosm{m}_{i}+\microcosm{s}_{i+1}$ where $\microcosm{s_{i+1}}$ is the monoid generated by the single map: 
$$s_{i+1}:(n,(x_{1},x_{2},\dots))\mapsto(n,(x_{i+1},x_{2},\dots,x_{i},x_{1},x_{i+2},\dots))$$
\item $\microcosm{m}_{\infty}=\cup_{i\in\naturalN}\microcosm{m}_{i}$.
\end{itemize}
The intuition is that a microcosm $\microcosm{m}$ represents the set of computational principles available to write programs in the model. The operation $+$ thus extends the set of principles at disposal, increasing expressivity. As a consequence, the set of languages characterised by the type $\oc\ListTypeBin\multimap\cond{Bool}$ becomes larger and larger as we consider extensions of the microcosms. As an example, the microcosm $\microcosm{m}_{1}$ corresponds to allowing oneself to compute with automata. Expanding this microcosm by adding a map $s_{2}$ yields $\microcosm{m}_{2}=\microcosm{m}_{1}+\microcosm{s}_{2}$ and corresponds to the addition of a new computational principle: using a second head. 

\begin{theorem}\label{th1}
The model $\dmodel{\{1\}}{m_{\textnormal{i}}}$~$(i\in\naturalN\cup\{\infty\})$ characterises the class \cctwdfa{i}.
\end{theorem}

In particular, the model $\dmodel{\{1\}}{m_{\textnormal{1}}}$ characterises the class \Regular of Regular languages and the model $\dmodel{\{1\}}{m_{\infty}}$ characterises the class \Logspace.

\subsection{Examples}

\paragraph{Examples of words}
We here work two examples to illustrate how the representation of computation by graphings works out. First, we give the representation as graphings of the lists $\star 0$ (\autoref{fig0}), $\star 11$ (\autoref{fig11}) and $\star 01$ (\autoref{fig01}). In the illustrations, the vertices, e.g. '0i', '0o', represent disjoint segments of unit length, e.g. $[0,1]$, $[1,2]$. As mentioned in the caption, the plain edges are realised as translations. Notice that those edges are cut into pieces (2 pieces each for the list $\star 0$, and three pieces each for the others). This is due to the fact that these graphings represent exponentials of the word representation: the exponential splits the unit segment into as many pieces as the length of the list; intuitively each piece of the unit segments correspond to the \enquote{address} of each bit. The edges then explicit the ordering: each symbol comes with has pointers to its preceding and succeeding symbols. E.g. in the case of the first '1' in the list $\star 11$, there is an edge from '1i' to '$\star$o' representing the fact that the preceding symbol was '$\star$', and there is an edge from '1o' to '1i' representing the fact that the next symbol is a '1'; notice moreover that these edges move between the corresponding addresses.

\paragraph{A first machine}
As a first example, let us represent an automata that accepts a word $\word{w}$ if and only if $\word{w}$ contains at least one symbol '1'. This example was chosen because of the simplicity of the representation. Indeed, one can represent it by a graphing with a single state and which uses a single head (hence it is a graphing in the microcosm $\microcosm{m}_{1}$). Its transition relation is then defined as follows: if the symbol read is a '1' then stop, if it is a '0', move the head to read the next symbol, if it is a '$\star$' then reject. The representation as a graphing is shown in \autoref{fig1machine}. Notice that the symbol '$\star$' plays both the role of left and right end-markers. 

We then represent the computation of this graphing with the representations of the words $\star 0$ (\autoref{fig0}), $\star 11$ (\autoref{fig11}) and $\star 01$ (\autoref{fig01}) in \autoref{fig1machinecomp}. The computation is illustrated by showing the two graphings (the machine and the integer), one on each side of the set of vertices. Notice that in those figures the graphing representing the machine has been replaced by one of its \emph{refinements} to reflect the splitting of the unit intervals appearing in the representation of the words. The result of the computation is the set of alternating paths from the set of vertices $\{\text{accept},\text{reject}\}$ to itself: in each case there is at most one, which is stressed by drawing the edges it is composed in boldface. The machine accepts the input if and only if there is a single path from "accept" to "accept". One can see in the figure that this machine accepts the words $\star 11$ and $\star 01$ but not the word $\star 0$.

Notice how the computation can be understood as a game between two players. We illustrate this on the computation on the word $\star 01$. First the machine, through its edge from "accept" to '$\star$o', asks the input \enquote{What's your first symbol?}. The integer, through its edge from '$\star$o' to '0i', answers \enquote{It's a '0'.}. Then the machine asks \enquote{What's your next symbol?} (the edge from '0i' to '0o'), to which the integer replies \enquote{It's a '1'.} (the edge from '0o' to '1i'). At this point, the machine accepts (the edge from '1i' to "accept").

\paragraph{A second machine} We chose as a second example a more complex machine. The language it accepts is actually a regular language (the words that contain at least one symbol '1' and one symbol '0'), i.e. can be computed by a single-head automata. However, we chose to compute this language by first looking for a symbol '1' in the word and, if this part is successful, activate a second head looking for a symbol '0', i.e. we compute it with a graphing in the microcosm $\mathfrak{m}_{2}$ even if as a regular language it can be computed by graphings in $\microcosm{m}_{1}$. This automata has two states corresponding to the two parts in the algorithmic procedure: a state \enquote{looking for a '1'} and a state \enquote{looking for a '0'}; we write these states "L1" and "L0" respectively. The graphing representing this automata is represented in \autoref{fig10machine}, where there are two sets of vertices one above the other: the row below shows the vertices in state L1 while the row on top shows the vertices in state L2. Notice the two dashed lines which are realised using a permutation and correspond to the change of principal head\footnote{The machine represented by graphings have a single active head at a given time. The permutations $s_{i}$ ($i\leqslant 2$) swap the $i$-th head with the first one, making the $i$-th active. Reusing the permutation $s_{i}$ then activates the former active head and \enquote{deactivates} the $i$-th head.}.

We represented the computation of this machine with the same inputs $\star 0$, $\star 11$, and $\star 01$ in \autoref{fig10machinecomp}. Notice that to deal with states, the inputs are duplicated. Once again, acceptance corresponds with the existence of a path from the "accept" vertex to itself (without changing states). The first computation does not go far: the first head moves along the word and then stops as it did not encounter a symbol '1' (this path is shown by thick edges). 

The second computation is more interesting. First, it follows the computation of the previous machine on the same input: it asks for the first symbol, acknowledge that it is a '1', then activates the second head and changes state. At this point the path continues as the second head moves along the input: this seems to contradict the picture since we our path seems to arrive on the middle splitting of the vertex '$\star$o' which would not allow us to continue through the input edge whose source is the left-hand splitting of this same vertex. However, this is forgetting that we changed the active head. Indeed, what happens in fact is that the permutation moves the splitting along a new direction: while the same transition realised as a simple translation would forbid us to continue the computation, realising it by the permutation $s_{2}$ allows the computation to continue. We illustrate in \autoref{illustratepermutations} how the permutation actually interact with the splitting, in the case the interval is split into two pieces. The representation of vertices as unit intervals is no longer sound as we are actually using the two first copies of $[0,1]$ in the Hilbert cube $[0,1]^{\naturalN}$ to represent this computation, hence working with unit squares. In this case, however, the second head moves along the input without encountering a symbol '0' and then stops. This path is shown by thick edges in the figure.

The last computation accepts the input. As in the previous case, the first head moves along the input, encounters a symbol '1', changes its state and activates the second head. As in the previous case, the computation can continue at this point, and the second head encounters a '0' as the first symbol of the input. Then the machine accepts. This path is shown in the figure by thick edges.

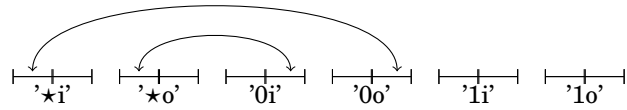
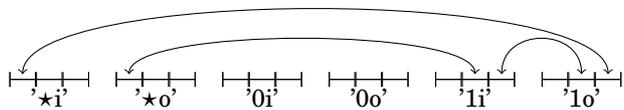
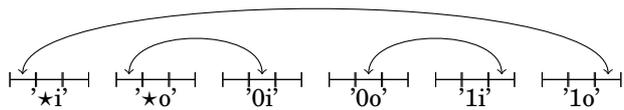
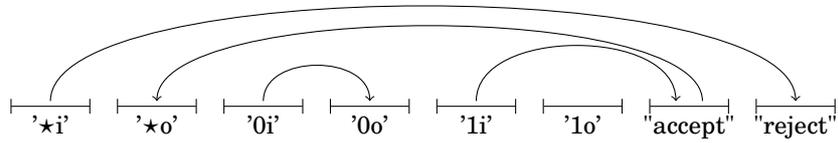
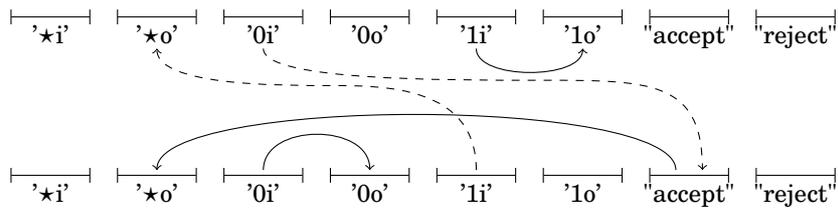
\begin{figure}
\centering
\subfloat[The graphing representing the word $\star 0$.\label{fig0}]{
\begin{tikzpicture}[x=0.7cm,y=0.7cm]
	\draw[|-|] (0,0) -- (1.5,0) node [midway,below] {'$\star$i'};
		\draw[|-|] (0,0) -- (0.75,0) {};
		\draw[|-|] (0.75,0) -- (1.5,0) {};
	\draw[|-|] (2,0) -- (3.5,0) node [midway,below] {'$\star$o'};
		\draw[|-|] (2,0) -- (2.75,0) {};
		\draw[|-|] (2.75,0) -- (3.5,0) {};
	\draw[|-|] (4,0) -- (5.5,0) node [midway,below] {'0i'};
		\draw[|-|] (4,0) -- (4.75,0) {};
		\draw[|-|] (4.75,0) -- (5.5,0) {};
	\draw[|-|] (6,0) -- (7.5,0) node [midway,below] {'0o'};
		\draw[|-|] (6,0) -- (6.75,0) {};
		\draw[|-|] (6.75,0) -- (7.5,0) {};
	\draw[|-|] (8,0) -- (9.5,0) node [midway,below] {'1i'};
		\draw[|-|] (8,0) -- (8.75,0) {};
		\draw[|-|] (8.75,0) -- (9.5,0) {};
	\draw[|-|] (10,0) -- (11.5,0) node [midway,below] {'1o'};
		\draw[|-|] (10,0) -- (10.75,0) {};
		\draw[|-|] (10.75,0) -- (11.5,0) {};
		
	\draw[<->] (2.375,0.1) .. controls (2.375,1) and (5.225,1) .. (5.225,0.1) {};
	\draw[<->] (7.225,0.1) .. controls (7.225,1.75) and (0.375,1.75) .. (0.375,0.1) {};
\end{tikzpicture}
}

\subfloat[The graphing representing the word $\star 11$.\label{fig11}]{\begin{tikzpicture}[x=0.7cm,y=0.7cm]
	\draw[|-|] (0,0) -- (1.5,0) node [midway,below] {'$\star$i'};
		\draw[|-|] (0,0) -- (0.5,0) {};
		\draw[|-|] (0.5,0) -- (1,0) {};
		\draw[|-|] (1,0) -- (1.5,0) {};
	\draw[|-|] (2,0) -- (3.5,0) node [midway,below] {'$\star$o'};
		\draw[|-|] (2,0) -- (2.5,0) {};
		\draw[|-|] (2.5,0) -- (3,0) {};
		\draw[|-|] (3,0) -- (3.5,0) {};
	\draw[|-|] (4,0) -- (5.5,0) node [midway,below] {'0i'};
		\draw[|-|] (4,0) -- (4.5,0) {};
		\draw[|-|] (4.5,0) -- (5,0) {};
		\draw[|-|] (5,0) -- (5.5,0) {};
	\draw[|-|] (6,0) -- (7.5,0) node [midway,below] {'0o'};
		\draw[|-|] (6,0) -- (6.5,0) {};
		\draw[|-|] (6.5,0) -- (7,0) {};
		\draw[|-|] (7,0) -- (7.5,0) {};
	\draw[|-|] (8,0) -- (9.5,0) node [midway,below] {'1i'};
		\draw[|-|] (8,0) -- (8.5,0) {};
		\draw[|-|] (8.5,0) -- (9,0) {};
		\draw[|-|] (9,0) -- (9.5,0) {};
	\draw[|-|] (10,0) -- (11.5,0) node [midway,below] {'1o'};
		\draw[|-|] (10,0) -- (10.5,0) {};
		\draw[|-|] (10.5,0) -- (11,0) {};
		\draw[|-|] (11,0) -- (11.5,0) {};
		
	\draw[<->] (2.25,0.1) .. controls (2.25,1) and (8.75,1) .. (8.75,0.1) {};
	\draw[<->] (10.75,0.1) .. controls (10.75,1) and (9.25,1) .. (9.25,0.1) {};
	\draw[<->] (0.25,0.1) .. controls (0.25,1.75) and (11.25,1.75) .. (11.25,0.1) {};
\end{tikzpicture}
}

\subfloat[The graphing representing the word $\star 01$.\label{fig01}]{\begin{tikzpicture}[x=0.7cm,y=0.7cm]
	\draw[|-|] (0,0) -- (1.5,0) node [midway,below] {'$\star$i'};
		\draw[|-|] (0,0) -- (0.5,0) {};
		\draw[|-|] (0.5,0) -- (1,0) {};
		\draw[|-|] (1,0) -- (1.5,0) {};
	\draw[|-|] (2,0) -- (3.5,0) node [midway,below] {'$\star$o'};
		\draw[|-|] (2,0) -- (2.5,0) {};
		\draw[|-|] (2.5,0) -- (3,0) {};
		\draw[|-|] (3,0) -- (3.5,0) {};
	\draw[|-|] (4,0) -- (5.5,0) node [midway,below] {'0i'};
		\draw[|-|] (4,0) -- (4.5,0) {};
		\draw[|-|] (4.5,0) -- (5,0) {};
		\draw[|-|] (5,0) -- (5.5,0) {};
	\draw[|-|] (6,0) -- (7.5,0) node [midway,below] {'0o'};
		\draw[|-|] (6,0) -- (6.5,0) {};
		\draw[|-|] (6.5,0) -- (7,0) {};
		\draw[|-|] (7,0) -- (7.5,0) {};
	\draw[|-|] (8,0) -- (9.5,0) node [midway,below] {'1i'};
		\draw[|-|] (8,0) -- (8.5,0) {};
		\draw[|-|] (8.5,0) -- (9,0) {};
		\draw[|-|] (9,0) -- (9.5,0) {};
	\draw[|-|] (10,0) -- (11.5,0) node [midway,below] {'1o'};
		\draw[|-|] (10,0) -- (10.5,0) {};
		\draw[|-|] (10.5,0) -- (11,0) {};
		\draw[|-|] (11,0) -- (11.5,0) {};
		
	\draw[<->] (2.25,0.1) .. controls (2.25,1) and (4.75,1) .. (4.75,0.1) {};
	\draw[<->] (6.75,0.1) .. controls (6.75,1) and (9.25,1) .. (9.25,0.1) {};
	\draw[<->] (0.25,0.1) .. controls (0.25,1.75) and (11.25,1.75) .. (11.25,0.1) {};
\end{tikzpicture}
}

\subfloat[The graphing representing the '1'-machine.\label{fig1machine}]{
\begin{tikzpicture}[x=0.7cm,y=0.7cm]
	\draw[|-|] (0,0) -- (1.5,0) node [midway,below] {'$\star$i'};
	\draw[|-|] (2,0) -- (3.5,0) node [midway,below] {'$\star$o'};
	\draw[|-|] (4,0) -- (5.5,0) node [midway,below] {'0i'};
	\draw[|-|] (6,0) -- (7.5,0) node [midway,below] {'0o'};
	\draw[|-|] (8,0) -- (9.5,0) node [midway,below] {'1i'};
	\draw[|-|] (10,0) -- (11.5,0) node [midway,below] {'1o'};
	\draw[|-|] (12,0) -- (13.5,0) node [midway,below] {"accept"};
	\draw[|-|] (14,0) -- (15.5,0) node [midway,below] {"reject"};
	
	\draw[->] (13,0.1) .. controls (13,2) and (2.75,2) .. (2.75,0.1) {};
	\draw[->] (0.75,0.1) .. controls (0.75,2.5) and (14.75,2.5) .. (14.75,0.1) {};
	\draw[->] (4.75,0.1) .. controls (4.75,1) and (6.75,1) .. (6.75,0.1) {};
	\draw[->] (8.75,0.1) .. controls (8.75,1.5) and (12.5,1.5) .. (12.5,0.1) {};
\end{tikzpicture}
}

\subfloat[The graphing representing the '1'-machine.\label{fig10machine}]{
\begin{tikzpicture}[x=0.7cm,y=0.7cm]
	\draw[|-|] (0,0) -- (1.5,0) node [midway,below] {'$\star$i'};
	\draw[|-|] (2,0) -- (3.5,0) node [midway,below] {'$\star$o'};
	\draw[|-|] (4,0) -- (5.5,0) node [midway,below] {'0i'};
	\draw[|-|] (6,0) -- (7.5,0) node [midway,below] {'0o'};
	\draw[|-|] (8,0) -- (9.5,0) node [midway,below] {'1i'};
	\draw[|-|] (10,0) -- (11.5,0) node [midway,below] {'1o'};
	\draw[|-|] (12,0) -- (13.5,0) node [midway,below] {"accept"};
	\draw[|-|] (14,0) -- (15.5,0) node [midway,below] {"reject"};
	
	\draw[|-|] (0,3) -- (1.5,3) node [midway,below] {'$\star$i'};
	\draw[|-|] (2,3) -- (3.5,3) node [midway,below] {'$\star$o'};
	\draw[|-|] (4,3) -- (5.5,3) node [midway,below] {'0i'};
	\draw[|-|] (6,3) -- (7.5,3) node [midway,below] {'0o'};
	\draw[|-|] (8,3) -- (9.5,3) node [midway,below] {'1i'};
	\draw[|-|] (10,3) -- (11.5,3) node [midway,below] {'1o'};
	\draw[|-|] (12,3) -- (13.5,3) node [midway,below] {"accept"};
	\draw[|-|] (14,3) -- (15.5,3) node [midway,below] {"reject"};
	
	\draw[->] (12.5,0.1) .. controls (12.5,1.5) and (2.75,1.5) .. (2.75,0.1) {};
	\draw[->] (4.75,0.1) .. controls (4.75,1) and (6.75,1) .. (6.75,0.1) {};
	\draw[->,dashed] (8.75,0.1) .. controls (8.75,1.7) and (7,1.7) .. (6,1.7) .. controls (5,1.7) and (2.75,1.7) .. (2.75,2.4) {};
	
	\draw[->,dashed] (4.75,2.4) .. controls (4.75,1.8) and (8,1.8) .. (10,1.8) .. controls (12,1.8) and (13,1.8) .. (13,0.1) {};
	\draw[->] (8.75,2.4) .. controls (8.75,1.8) and (10.75,1.8) .. (10.75,2.4) {};

\end{tikzpicture}
}
\caption{Examples: the plain lines are realised by simple translations, while the dashed lines are realised by a composition of a translation and the map $s_{2}$.}
\end{figure}

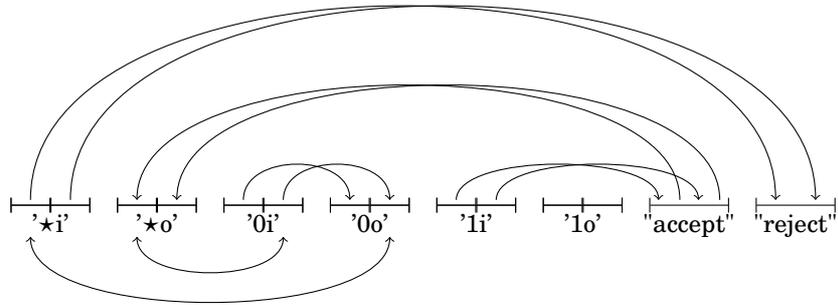
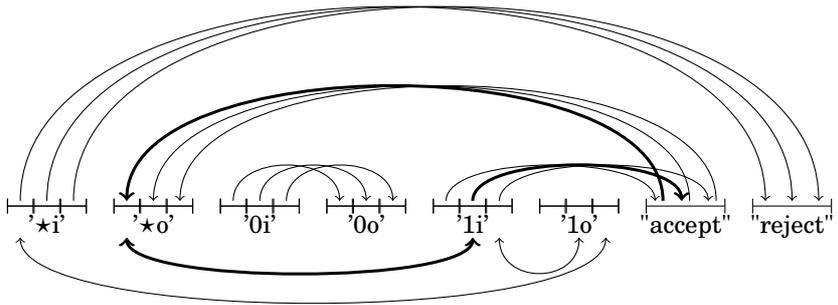
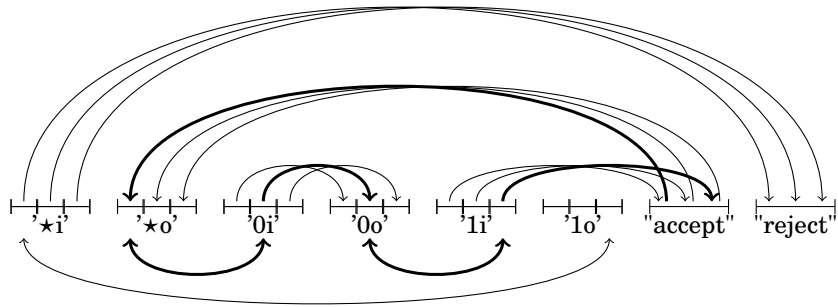
\begin{figure}
\centering
\subfloat[Computing with input $\star 0$.\label{fig1machine0}]{
\begin{tikzpicture}[x=0.7cm,y=0.7cm]
	\draw[|-|] (0,0) -- (1.5,0) node [midway,below] {'$\star$i'};
		\draw[|-|] (0,0) -- (0.75,0) {};
		\draw[|-|] (0.75,0) -- (1.5,0) {};
	\draw[|-|] (2,0) -- (3.5,0) node [midway,below] {'$\star$o'};
		\draw[|-|] (2,0) -- (2.75,0) {};
		\draw[|-|] (2.75,0) -- (3.5,0) {};
	\draw[|-|] (4,0) -- (5.5,0) node [midway,below] {'0i'};
		\draw[|-|] (4,0) -- (4.75,0) {};
		\draw[|-|] (4.75,0) -- (5.5,0) {};
	\draw[|-|] (6,0) -- (7.5,0) node [midway,below] {'0o'};
		\draw[|-|] (6,0) -- (6.75,0) {};
		\draw[|-|] (6.75,0) -- (7.5,0) {};
	\draw[|-|] (8,0) -- (9.5,0) node [midway,below] {'1i'};
		\draw[|-|] (8,0) -- (8.75,0) {};
		\draw[|-|] (8.75,0) -- (9.5,0) {};
	\draw[|-|] (10,0) -- (11.5,0) node [midway,below] {'1o'};
		\draw[|-|] (10,0) -- (10.75,0) {};
		\draw[|-|] (10.75,0) -- (11.5,0) {};
		
	\draw[<->] (2.375,-0.6) .. controls (2.375,-1.5) and (5.125,-1.5) .. (5.125,-0.6) {};
	\draw[<->] (7.125,-0.6) .. controls (7.125,-2.25) and (0.375,-2.25) .. (0.375,-0.6) {};
	
	\draw[|-|] (12,0) -- (13.5,0) node [midway,below] {"accept"};
	\draw[|-|] (14,0) -- (15.5,0) node [midway,below] {"reject"};
	
	\draw[->] (12.575,0.1) .. controls (12.575,3) and (2.375,3) .. (2.375,0.1) {};
	\draw[->] (13.325,0.1) .. controls (13.325,3) and (3.125,3) .. (3.125,0.1) {};

	\draw[->] (0.375,0.1) .. controls (0.375,5) and (14.375,5) .. (14.375,0.1) {};
	\draw[->] (1.125,0.1) .. controls (1.125,5) and (15.125,5) .. (15.125,0.1) {};

	\draw[->] (4.375,0.1) .. controls (4.375,1) and (6.375,1) .. (6.375,0.1) {};
	\draw[->] (5.125,0.1) .. controls (5.125,1) and (7.125,1) .. (7.125,0.1) {};

	\draw[->] (8.375,0.1) .. controls (8.375,1) and (12.175,1) .. (12.175,0.1) {};
	\draw[->] (9.125,0.1) .. controls (9.125,1) and (12.925,1) .. (12.925,0.1) {};

\end{tikzpicture}
}

\subfloat[Computing with input $\star 11$.\label{fig1machine11}]{\begin{tikzpicture}[x=0.7cm,y=0.7cm]
	\draw[|-|] (0,0) -- (1.5,0) node [midway,below] {'$\star$i'};
		\draw[|-|] (0,0) -- (0.5,0) {};
		\draw[|-|] (0.5,0) -- (1,0) {};
		\draw[|-|] (1,0) -- (1.5,0) {};
	\draw[|-|] (2,0) -- (3.5,0) node [midway,below] {'$\star$o'};
		\draw[|-|] (2,0) -- (2.5,0) {};
		\draw[|-|] (2.5,0) -- (3,0) {};
		\draw[|-|] (3,0) -- (3.5,0) {};
	\draw[|-|] (4,0) -- (5.5,0) node [midway,below] {'0i'};
		\draw[|-|] (4,0) -- (4.5,0) {};
		\draw[|-|] (4.5,0) -- (5,0) {};
		\draw[|-|] (5,0) -- (5.5,0) {};
	\draw[|-|] (6,0) -- (7.5,0) node [midway,below] {'0o'};
		\draw[|-|] (6,0) -- (6.5,0) {};
		\draw[|-|] (6.5,0) -- (7,0) {};
		\draw[|-|] (7,0) -- (7.5,0) {};
	\draw[|-|] (8,0) -- (9.5,0) node [midway,below] {'1i'};
		\draw[|-|] (8,0) -- (8.5,0) {};
		\draw[|-|] (8.5,0) -- (9,0) {};
		\draw[|-|] (9,0) -- (9.5,0) {};
	\draw[|-|] (10,0) -- (11.5,0) node [midway,below] {'1o'};
		\draw[|-|] (10,0) -- (10.5,0) {};
		\draw[|-|] (10.5,0) -- (11,0) {};
		\draw[|-|] (11,0) -- (11.5,0) {};
		
	\draw[<->,line width=0.4mm] (2.25,-0.6) .. controls (2.25,-1.5) and (8.75,-1.5) .. (8.75,-0.6) {};
	\draw[<->] (10.75,-0.6) .. controls (10.75,-1.5) and (9.25,-1.5) .. (9.25,-0.6) {};
	\draw[<->] (0.25,-0.6) .. controls (0.25,-2.25) and (11.25,-2.25) .. (11.25,-0.6) {};
	
	\draw[|-|] (12,0) -- (13.5,0) node [midway,below] {"accept"};
	\draw[|-|] (14,0) -- (15.5,0) node [midway,below] {"reject"};
	
	\draw[->,line width=0.4mm] (12.325,0.1) .. controls (12.325,3) and (2.25,3) .. (2.25,0.1) {};
	\draw[->] (12.825,0.1) .. controls (12.825,3) and (2.75,3) .. (2.75,0.1) {};
	\draw[->] (13.325,0.1) .. controls (13.325,3) and (3.25,3) .. (3.25,0.1) {};

	\draw[->] (0.25,0.1) .. controls (0.25,5) and (14.25,5) .. (14.25,0.1) {};
	\draw[->] (0.75,0.1) .. controls (0.75,5) and (14.75,5) .. (14.75,0.1) {};
	\draw[->] (1.25,0.1) .. controls (1.25,5) and (15.25,5) .. (15.25,0.1) {};
	
	\draw[->] (4.25,0.1) .. controls (4.25,1) and (6.25,1) .. (6.25,0.1) {};
	\draw[->] (4.75,0.1) .. controls (4.75,1) and (6.75,1) .. (6.75,0.1) {};
	\draw[->] (5.25,0.1) .. controls (5.25,1) and (7.25,1) .. (7.25,0.1) {};
	
	\draw[->] (8.25,0.1) .. controls (8.25,1) and (12.175,1) .. (12.175,0.1) {};
	\draw[->,line width=0.4mm] (8.75,0.1) .. controls (8.75,1) and (12.675,1) .. (12.675,0.1) {};
	\draw[->] (9.25,0.1) .. controls (9.25,1) and (13.175,1) .. (13.175,0.1) {};
\end{tikzpicture}
}

\subfloat[Computing with input $\star 01$.\label{fig1machine01}]{
\begin{tikzpicture}[x=0.7cm,y=0.7cm]
	\draw[|-|] (0,0) -- (1.5,0) node [midway,below] {'$\star$i'};
		\draw[|-|] (0,0) -- (0.5,0) {};
		\draw[|-|] (0.5,0) -- (1,0) {};
		\draw[|-|] (1,0) -- (1.5,0) {};
	\draw[|-|] (2,0) -- (3.5,0) node [midway,below] {'$\star$o'};
		\draw[|-|] (2,0) -- (2.5,0) {};
		\draw[|-|] (2.5,0) -- (3,0) {};
		\draw[|-|] (3,0) -- (3.5,0) {};
	\draw[|-|] (4,0) -- (5.5,0) node [midway,below] {'0i'};
		\draw[|-|] (4,0) -- (4.5,0) {};
		\draw[|-|] (4.5,0) -- (5,0) {};
		\draw[|-|] (5,0) -- (5.5,0) {};
	\draw[|-|] (6,0) -- (7.5,0) node [midway,below] {'0o'};
		\draw[|-|] (6,0) -- (6.5,0) {};
		\draw[|-|] (6.5,0) -- (7,0) {};
		\draw[|-|] (7,0) -- (7.5,0) {};
	\draw[|-|] (8,0) -- (9.5,0) node [midway,below] {'1i'};
		\draw[|-|] (8,0) -- (8.5,0) {};
		\draw[|-|] (8.5,0) -- (9,0) {};
		\draw[|-|] (9,0) -- (9.5,0) {};
	\draw[|-|] (10,0) -- (11.5,0) node [midway,below] {'1o'};
		\draw[|-|] (10,0) -- (10.5,0) {};
		\draw[|-|] (10.5,0) -- (11,0) {};
		\draw[|-|] (11,0) -- (11.5,0) {};
	\draw[|-|] (12,0) -- (13.5,0) node [midway,below] {"accept"};
	\draw[|-|] (14,0) -- (15.5,0) node [midway,below] {"reject"};
	
	\draw[->,line width=0.4mm] (12.325,0.1) .. controls (12.325,3) and (2.25,3) .. (2.25,0.1) {};
	\draw[->] (12.825,0.1) .. controls (12.825,3) and (2.75,3) .. (2.75,0.1) {};
	\draw[->] (13.325,0.1) .. controls (13.325,3) and (3.25,3) .. (3.25,0.1) {};

	\draw[->] (0.25,0.1) .. controls (0.25,5) and (14.25,5) .. (14.25,0.1) {};
	\draw[->] (0.75,0.1) .. controls (0.75,5) and (14.75,5) .. (14.75,0.1) {};
	\draw[->] (1.25,0.1) .. controls (1.25,5) and (15.25,5) .. (15.25,0.1) {};
	
	\draw[->] (4.25,0.1) .. controls (4.25,1) and (6.25,1) .. (6.25,0.1) {};
	\draw[->,line width=0.4mm] (4.75,0.1) .. controls (4.75,1) and (6.75,1) .. (6.75,0.1) {};
	\draw[->] (5.25,0.1) .. controls (5.25,1) and (7.25,1) .. (7.25,0.1) {};
	
	\draw[->] (8.25,0.1) .. controls (8.25,1) and (12.175,1) .. (12.175,0.1) {};
	\draw[->] (8.75,0.1) .. controls (8.75,1) and (12.675,1) .. (12.675,0.1) {};
	\draw[->,line width=0.4mm] (9.25,0.1) .. controls (9.25,1) and (13.175,1) .. (13.175,0.1) {};

	\draw[<->,line width=0.4mm] (2.25,-0.6) .. controls (2.25,-1.5) and (4.75,-1.5) .. (4.75,-0.6) {};
	\draw[<->,line width=0.4mm] (6.75,-0.6) .. controls (6.75,-1.5) and (9.25,-1.5) .. (9.25,-0.6) {};
	\draw[<->] (0.25,-0.6) .. controls (0.25,-2.25) and (11.25,-2.25) .. (11.25,-0.6) {};
\end{tikzpicture}
}
\caption{Example: computing with the first machine.}\label{fig1machinecomp}
\end{figure}

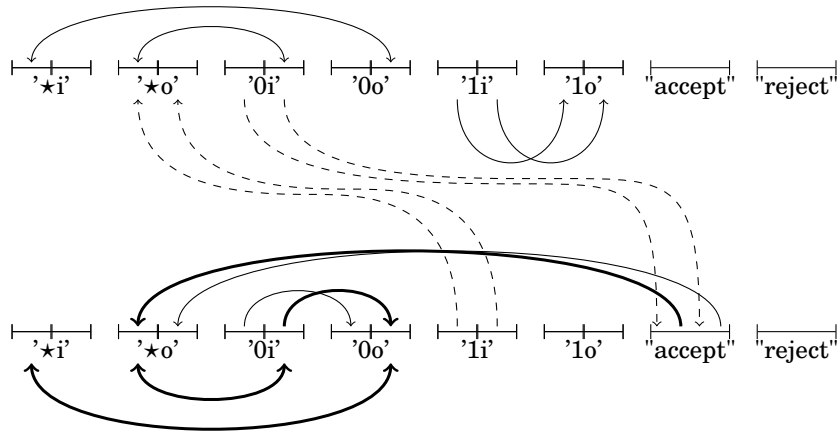
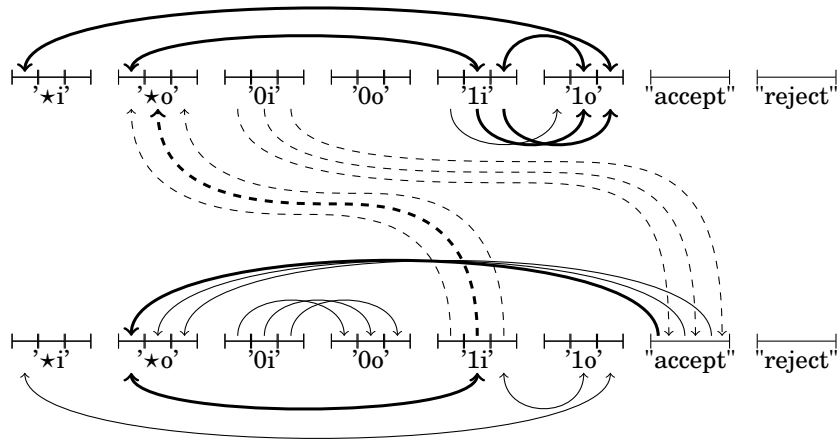
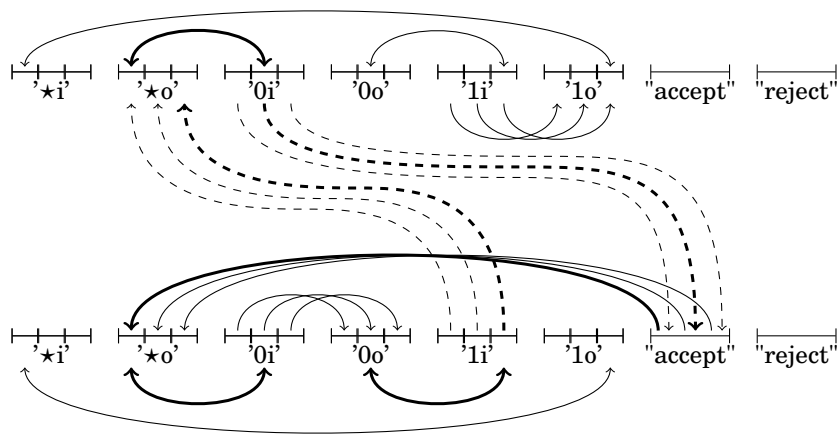
\begin{figure}
\centering
\subfloat[Computing with input $\star 0$.\label{fig10machine0}]{
\begin{tikzpicture}[x=0.7cm,y=0.7cm]
	\draw[|-|] (0,0) -- (1.5,0) node [midway,below] {'$\star$i'};
		\draw[|-|] (0,0) -- (0.75,0) {};
		\draw[|-|] (0.75,0) -- (1.5,0) {};
	\draw[|-|] (2,0) -- (3.5,0) node [midway,below] {'$\star$o'};
		\draw[|-|] (2,0) -- (2.75,0) {};
		\draw[|-|] (2.75,0) -- (3.5,0) {};
	\draw[|-|] (4,0) -- (5.5,0) node [midway,below] {'0i'};
		\draw[|-|] (4,0) -- (4.75,0) {};
		\draw[|-|] (4.75,0) -- (5.5,0) {};
	\draw[|-|] (6,0) -- (7.5,0) node [midway,below] {'0o'};
		\draw[|-|] (6,0) -- (6.75,0) {};
		\draw[|-|] (6.75,0) -- (7.5,0) {};
	\draw[|-|] (8,0) -- (9.5,0) node [midway,below] {'1i'};
		\draw[|-|] (8,0) -- (8.75,0) {};
		\draw[|-|] (8.75,0) -- (9.5,0) {};
	\draw[|-|] (10,0) -- (11.5,0) node [midway,below] {'1o'};
		\draw[|-|] (10,0) -- (10.75,0) {};
		\draw[|-|] (10.75,0) -- (11.5,0) {};
		
	\draw[<->,line width=0.4mm] (2.375,-0.6) .. controls (2.375,-1.5) and (5.125,-1.5) .. (5.125,-0.6) {};
	\draw[<->,line width=0.4mm] (7.125,-0.6) .. controls (7.125,-2.25) and (0.375,-2.25) .. (0.375,-0.6) {};
	
	\draw[|-|] (0,5) -- (1.5,5) node [midway,below] {'$\star$i'};
		\draw[|-|] (0,5) -- (0.75,5) {};
		\draw[|-|] (0.75,5) -- (1.5,5) {};
	\draw[|-|] (2,5) -- (3.5,5) node [midway,below] {'$\star$o'};
		\draw[|-|] (2,5) -- (2.75,5) {};
		\draw[|-|] (2.75,5) -- (3.5,5) {};
	\draw[|-|] (4,5) -- (5.5,5) node [midway,below] {'0i'};
		\draw[|-|] (4,5) -- (4.75,5) {};
		\draw[|-|] (4.75,5) -- (5.5,5) {};
	\draw[|-|] (6,5) -- (7.5,5) node [midway,below] {'0o'};
		\draw[|-|] (6,5) -- (6.75,5) {};
		\draw[|-|] (6.75,5) -- (7.5,5) {};
	\draw[|-|] (8,5) -- (9.5,5) node [midway,below] {'1i'};
		\draw[|-|] (8,5) -- (8.75,5) {};
		\draw[|-|] (8.75,5) -- (9.5,5) {};
	\draw[|-|] (10,5) -- (11.5,5) node [midway,below] {'1o'};
		\draw[|-|] (10,5) -- (10.75,5) {};
		\draw[|-|] (10.75,5) -- (11.5,5) {};
		
	\draw[<->] (2.375,5.1) .. controls (2.375,6) and (5.125,6) .. (5.125,5.1) {};
	\draw[<->] (7.125,5.1) .. controls (7.125,6.5) and (0.375,6.5) .. (0.375,5.1) {};
	
	\draw[|-|] (12,0) -- (13.5,0) node [midway,below] {"accept"};
	\draw[|-|] (14,0) -- (15.5,0) node [midway,below] {"reject"};
	
	\draw[|-|] (12,5) -- (13.5,5) node [midway,below] {"accept"};
	\draw[|-|] (14,5) -- (15.5,5) node [midway,below] {"reject"};

	\draw[->,line width=0.4mm] (12.575,0.1) .. controls (12.575,2) and (2.375,2) .. (2.375,0.1) {};
	\draw[->] (13.325,0.1) .. controls (13.325,2) and (3.125,2) .. (3.125,0.1) {};

	\draw[->] (4.375,0.1) .. controls (4.375,1) and (6.375,1) .. (6.375,0.1) {};
	\draw[->,line width=0.4mm] (5.125,0.1) .. controls (5.125,1) and (7.125,1) .. (7.125,0.1) {};
	
	\draw[->,dashed] (8.375,0.1) .. controls (8.375,2.6) and (6.625,2.6) .. (5.625,2.6) .. controls (4.625,2.6) and (2.375,2.6) .. (2.375,4.4) {};
	\draw[->,dashed] (9.125,0.1) .. controls (9.125,2.7) and (7.375,2.7) .. (6.375,2.7) .. controls (5.375,2.7) and (3.125,2.7) .. (3.125,4.4) {};
	
	\draw[->,dashed] (4.375,4.4) .. controls (4.375,2.8) and (7.625,2.8) .. (9.625,2.8) .. controls (11.625,2.8) and (12.125,2.8) .. (12.125,0.1) {};
	\draw[->,dashed] (5.125,4.4) .. controls (5.125,2.8) and (8.375,2.9) .. (10.375,2.9) .. controls (12.375,2.9) and (12.925,2.9) .. (12.925,0.1) {};	
	
	\draw[->] (8.375,4.4) .. controls (8.375,2.8) and (10.375,2.8) .. (10.375,4.4) {};
	\draw[->] (9.125,4.4) .. controls (9.125,2.8) and (11.125,2.8) .. (11.125,4.4) {};
\end{tikzpicture}
}

\subfloat[Computing with input $\star 11$.\label{fig10machine11}]{
\begin{tikzpicture}[x=0.7cm,y=0.7cm]
	\draw[|-|] (0,0) -- (1.5,0) node [midway,below] {'$\star$i'};
		\draw[|-|] (0,0) -- (0.5,0) {};
		\draw[|-|] (0.5,0) -- (1,0) {};
		\draw[|-|] (1,0) -- (1.5,0) {};
	\draw[|-|] (2,0) -- (3.5,0) node [midway,below] {'$\star$o'};
		\draw[|-|] (2,0) -- (2.5,0) {};
		\draw[|-|] (2.5,0) -- (3,0) {};
		\draw[|-|] (3,0) -- (3.5,0) {};
	\draw[|-|] (4,0) -- (5.5,0) node [midway,below] {'0i'};
		\draw[|-|] (4,0) -- (4.5,0) {};
		\draw[|-|] (4.5,0) -- (5,0) {};
		\draw[|-|] (5,0) -- (5.5,0) {};
	\draw[|-|] (6,0) -- (7.5,0) node [midway,below] {'0o'};
		\draw[|-|] (6,0) -- (6.5,0) {};
		\draw[|-|] (6.5,0) -- (7,0) {};
		\draw[|-|] (7,0) -- (7.5,0) {};
	\draw[|-|] (8,0) -- (9.5,0) node [midway,below] {'1i'};
		\draw[|-|] (8,0) -- (8.5,0) {};
		\draw[|-|] (8.5,0) -- (9,0) {};
		\draw[|-|] (9,0) -- (9.5,0) {};
	\draw[|-|] (10,0) -- (11.5,0) node [midway,below] {'1o'};
		\draw[|-|] (10,0) -- (10.5,0) {};
		\draw[|-|] (10.5,0) -- (11,0) {};
		\draw[|-|] (11,0) -- (11.5,0) {};
		
	\draw[<->,line width=0.4mm] (2.25,-0.6) .. controls (2.25,-1.5) and (8.75,-1.5) .. (8.75,-0.6) {};
	\draw[<->] (10.75,-0.6) .. controls (10.75,-1.5) and (9.25,-1.5) .. (9.25,-0.6) {};
	\draw[<->] (0.25,-0.6) .. controls (0.25,-2.25) and (11.25,-2.25) .. (11.25,-0.6) {};
	
	\draw[|-|] (0,5) -- (1.5,5) node [midway,below] {'$\star$i'};
		\draw[|-|] (0,5) -- (0.5,5) {};
		\draw[|-|] (0.5,5) -- (1,5) {};
		\draw[|-|] (1,5) -- (1.5,5) {};
	\draw[|-|] (2,5) -- (3.5,5) node [midway,below] {'$\star$o'};
		\draw[|-|] (2,5) -- (2.5,5) {};
		\draw[|-|] (2.5,5) -- (3,5) {};
		\draw[|-|] (3,5) -- (3.5,5) {};
	\draw[|-|] (4,5) -- (5.5,5) node [midway,below] {'0i'};
		\draw[|-|] (4,5) -- (4.5,5) {};
		\draw[|-|] (4.5,5) -- (5,5) {};
		\draw[|-|] (5,5) -- (5.5,5) {};
	\draw[|-|] (6,5) -- (7.5,5) node [midway,below] {'0o'};
		\draw[|-|] (6,5) -- (6.5,5) {};
		\draw[|-|] (6.5,5) -- (7,5) {};
		\draw[|-|] (7,5) -- (7.5,5) {};
	\draw[|-|] (8,5) -- (9.5,5) node [midway,below] {'1i'};
		\draw[|-|] (8,5) -- (8.5,5) {};
		\draw[|-|] (8.5,5) -- (9,5) {};
		\draw[|-|] (9,5) -- (9.5,5) {};
	\draw[|-|] (10,5) -- (11.5,5) node [midway,below] {'1o'};
		\draw[|-|] (10,5) -- (10.5,5) {};
		\draw[|-|] (10.5,5) -- (11,5) {};
		\draw[|-|] (11,5) -- (11.5,5) {};
		
	\draw[<->,line width=0.4mm] (2.25,5.1) .. controls (2.25,6) and (8.75,6) .. (8.75,5.1) {};
	\draw[<->,line width=0.4mm] (10.75,5.1) .. controls (10.75,6) and (9.25,6) .. (9.25,5.1) {};
	\draw[<->,line width=0.4mm] (0.25,5.1) .. controls (0.25,6.7) and (11.25,6.7) .. (11.25,5.1) {};
	
	\draw[|-|] (12,0) -- (13.5,0) node [midway,below] {"accept"};
	\draw[|-|] (14,0) -- (15.5,0) node [midway,below] {"reject"};
	
	\draw[|-|] (12,5) -- (13.5,5) node [midway,below] {"accept"};
	\draw[|-|] (14,5) -- (15.5,5) node [midway,below] {"reject"};

	\draw[->,line width=0.4mm] (12.15,0.1) .. controls (12.15,2) and (2.25,2) .. (2.25,0.1) {};
	\draw[->] (12.65,0.1) .. controls (12.65,2) and (2.75,2) .. (2.75,0.1) {};
	\draw[->] (13.15,0.1) .. controls (13.15,2) and (3.25,2) .. (3.25,0.1) {};

	\draw[->] (4.25,0.1) .. controls (4.25,1) and (6.25,1) .. (6.25,0.1) {};
	\draw[->] (4.75,0.1) .. controls (4.75,1) and (6.75,1) .. (6.75,0.1) {};
	\draw[->] (5.25,0.1) .. controls (5.25,1) and (7.25,1) .. (7.25,0.1) {};
	
	\draw[->,dashed] (8.25,0.1) .. controls (8.25,2.4) and (6.5,2.4) .. (5.5,2.4) .. controls (4.5,2.4) and (2.25,2.4) .. (2.25,4.4) {};
	\draw[->,dashed,line width=0.4mm] (8.75,0.1) .. controls (8.75,2.6) and (7,2.6) .. (6,2.6) .. controls (5,2.6) and (2.75,2.6) .. (2.75,4.4) {};
	\draw[->,dashed] (9.25,0.1) .. controls (9.25,2.8) and (7.5,2.8) .. (6.5,2.8) .. controls (5.5,2.8) and (3.25,2.8) .. (3.25,4.4) {};
	
	\draw[->,dashed] (4.25,4.4) .. controls (4.25,3) and (7.5,3) .. (9.5,3) .. controls (11.5,3) and (12.35,2.8) .. (12.35,0.1) {};
	\draw[->,dashed] (4.75,4.4) .. controls (4.75,3.2) and (7.5,3.2) .. (9.5,3.2) .. controls (11.5,3.2) and (12.85,3.2) .. (12.85,0.1) {};
	\draw[->,dashed] (5.25,4.4) .. controls (5.25,3.4) and (8,3.4) .. (10,3.4) .. controls (12,3.4) and (13.35,3.4) .. (13.35,0.1) {};	
	
	\draw[->] (8.25,4.4) .. controls (8.25,3.5) and (10.25,3.5) .. (10.25,4.4) {};
	\draw[->,line width=0.4mm] (8.75,4.4) .. controls (8.75,3.5) and (10.75,3.5) .. (10.75,4.4) {};
	\draw[->,line width=0.4mm] (9.25,4.4) .. controls (9.25,3.5) and (11.25,3.5) .. (11.25,4.4) {};
\end{tikzpicture}
}

\subfloat[Computing with input $\star 01$.\label{fig10machine01}]{
\begin{tikzpicture}[x=0.7cm,y=0.7cm]
	\draw[|-|] (0,0) -- (1.5,0) node [midway,below] {'$\star$i'};
		\draw[|-|] (0,0) -- (0.5,0) {};
		\draw[|-|] (0.5,0) -- (1,0) {};
		\draw[|-|] (1,0) -- (1.5,0) {};
	\draw[|-|] (2,0) -- (3.5,0) node [midway,below] {'$\star$o'};
		\draw[|-|] (2,0) -- (2.5,0) {};
		\draw[|-|] (2.5,0) -- (3,0) {};
		\draw[|-|] (3,0) -- (3.5,0) {};
	\draw[|-|] (4,0) -- (5.5,0) node [midway,below] {'0i'};
		\draw[|-|] (4,0) -- (4.5,0) {};
		\draw[|-|] (4.5,0) -- (5,0) {};
		\draw[|-|] (5,0) -- (5.5,0) {};
	\draw[|-|] (6,0) -- (7.5,0) node [midway,below] {'0o'};
		\draw[|-|] (6,0) -- (6.5,0) {};
		\draw[|-|] (6.5,0) -- (7,0) {};
		\draw[|-|] (7,0) -- (7.5,0) {};
	\draw[|-|] (8,0) -- (9.5,0) node [midway,below] {'1i'};
		\draw[|-|] (8,0) -- (8.5,0) {};
		\draw[|-|] (8.5,0) -- (9,0) {};
		\draw[|-|] (9,0) -- (9.5,0) {};
	\draw[|-|] (10,0) -- (11.5,0) node [midway,below] {'1o'};
		\draw[|-|] (10,0) -- (10.5,0) {};
		\draw[|-|] (10.5,0) -- (11,0) {};
		\draw[|-|] (11,0) -- (11.5,0) {};
		
	\draw[<->,line width=0.4mm] (2.25,-0.6) .. controls (2.25,-1.5) and (4.75,-1.5) .. (4.75,-0.6) {};
	\draw[<->,line width=0.4mm] (6.75,-0.6) .. controls (6.75,-1.5) and (9.25,-1.5) .. (9.25,-0.6) {};
	\draw[<->] (0.25,-0.6) .. controls (0.25,-2.25) and (11.25,-2.25) .. (11.25,-0.6) {};
	
	\draw[|-|] (0,5) -- (1.5,5) node [midway,below] {'$\star$i'};
		\draw[|-|] (0,5) -- (0.5,5) {};
		\draw[|-|] (0.5,5) -- (1,5) {};
		\draw[|-|] (1,5) -- (1.5,5) {};
	\draw[|-|] (2,5) -- (3.5,5) node [midway,below] {'$\star$o'};
		\draw[|-|] (2,5) -- (2.5,5) {};
		\draw[|-|] (2.5,5) -- (3,5) {};
		\draw[|-|] (3,5) -- (3.5,5) {};
	\draw[|-|] (4,5) -- (5.5,5) node [midway,below] {'0i'};
		\draw[|-|] (4,5) -- (4.5,5) {};
		\draw[|-|] (4.5,5) -- (5,5) {};
		\draw[|-|] (5,5) -- (5.5,5) {};
	\draw[|-|] (6,5) -- (7.5,5) node [midway,below] {'0o'};
		\draw[|-|] (6,5) -- (6.5,5) {};
		\draw[|-|] (6.5,5) -- (7,5) {};
		\draw[|-|] (7,5) -- (7.5,5) {};
	\draw[|-|] (8,5) -- (9.5,5) node [midway,below] {'1i'};
		\draw[|-|] (8,5) -- (8.5,5) {};
		\draw[|-|] (8.5,5) -- (9,5) {};
		\draw[|-|] (9,5) -- (9.5,5) {};
	\draw[|-|] (10,5) -- (11.5,5) node [midway,below] {'1o'};
		\draw[|-|] (10,5) -- (10.5,5) {};
		\draw[|-|] (10.5,5) -- (11,5) {};
		\draw[|-|] (11,5) -- (11.5,5) {};
		
	\draw[<->,line width=0.4mm] (2.25,5.1) .. controls (2.25,6) and (4.75,6) .. (4.75,5.1) {};
	\draw[<->] (6.75,5.1) .. controls (6.75,6) and (9.25,6) .. (9.25,5.1) {};
	\draw[<->] (0.25,5.1) .. controls (0.25,6.5) and (11.25,6.5) .. (11.25,5.1) {};
	
	\draw[|-|] (12,0) -- (13.5,0) node [midway,below] {"accept"};
	\draw[|-|] (14,0) -- (15.5,0) node [midway,below] {"reject"};
	
	\draw[|-|] (12,5) -- (13.5,5) node [midway,below] {"accept"};
	\draw[|-|] (14,5) -- (15.5,5) node [midway,below] {"reject"};

	\draw[->,line width=0.4mm] (12.15,0.1) .. controls (12.15,2) and (2.25,2) .. (2.25,0.1) {};
	\draw[->] (12.65,0.1) .. controls (12.65,2) and (2.75,2) .. (2.75,0.1) {};
	\draw[->] (13.15,0.1) .. controls (13.15,2) and (3.25,2) .. (3.25,0.1) {};

	\draw[->] (4.25,0.1) .. controls (4.25,1) and (6.25,1) .. (6.25,0.1) {};
	\draw[->] (4.75,0.1) .. controls (4.75,1) and (6.75,1) .. (6.75,0.1) {};
	\draw[->] (5.25,0.1) .. controls (5.25,1) and (7.25,1) .. (7.25,0.1) {};
	
	\draw[->,dashed] (8.25,0.1) .. controls (8.25,2.4) and (6.5,2.4) .. (5.5,2.4) .. controls (4.5,2.4) and (2.25,2.4) .. (2.25,4.4) {};
	\draw[->,dashed] (8.75,0.1) .. controls (8.75,2.6) and (7,2.6) .. (6,2.6) .. controls (5,2.6) and (2.75,2.6) .. (2.75,4.4) {};
	\draw[->,dashed,line width=0.4mm] (9.25,0.1) .. controls (9.25,2.8) and (7.5,2.8) .. (6.5,2.8) .. controls (5.5,2.8) and (3.25,2.8) .. (3.25,4.4) {};
	
	\draw[->,dashed] (4.25,4.4) .. controls (4.25,3) and (7.5,3) .. (9.5,3) .. controls (11.5,3) and (12.35,2.8) .. (12.35,0.1) {};
	\draw[->,dashed,line width=0.4mm] (4.75,4.4) .. controls (4.75,3.2) and (7.5,3.2) .. (9.5,3.2) .. controls (11.5,3.2) and (12.85,3.2) .. (12.85,0.1) {};
	\draw[->,dashed] (5.25,4.4) .. controls (5.25,3.4) and (8,3.4) .. (10,3.4) .. controls (12,3.4) and (13.35,3.4) .. (13.35,0.1) {};	
	
	\draw[->] (8.25,4.4) .. controls (8.25,3.5) and (10.25,3.5) .. (10.25,4.4) {};
	\draw[->] (8.75,4.4) .. controls (8.75,3.5) and (10.75,3.5) .. (10.75,4.4) {};
	\draw[->] (9.25,4.4) .. controls (9.25,3.5) and (11.25,3.5) .. (11.25,4.4) {};
\end{tikzpicture}
}
\caption{Example: computing with the second machine.}\label{fig10machinecomp}
\end{figure}

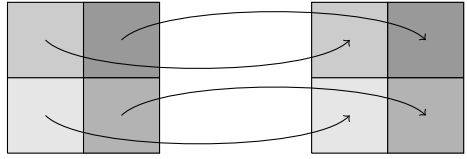
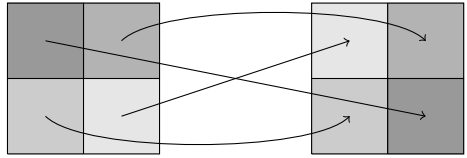
\begin{figure}
\centering
\subfloat[A simple translation]{
\begin{tikzpicture}[x={(0cm,1cm)},y={(1cm,0cm)}]
	\draw[-,fill,opacity=0.1] (0,0) -- (0,1) -- (1,1) -- (1,0) -- (0,0) {};
	\draw[-,fill,opacity=0.3] (0,1) -- (0,2) -- (1,2) -- (1,1) -- (0,1) {};
	\draw[-,fill,opacity=0.4] (1,1) -- (1,2) -- (2,2) -- (2,1) -- (1,1) {};
	\draw[-,fill,opacity=0.2] (1,0) -- (1,1) -- (2,1) -- (2,0) -- (1,0) {};
	\draw[-,fill,opacity=0.1] (0,4) -- (0,5) -- (1,5) -- (1,4) -- (0,4) {};
	\draw[-,fill,opacity=0.3] (0,5) -- (0,6) -- (1,6) -- (1,5) -- (0,5) {};
	\draw[-,fill,opacity=0.4] (1,5) -- (1,6) -- (2,6) -- (2,5) -- (1,5) {};
	\draw[-,fill,opacity=0.2] (1,4) -- (1,5) -- (2,5) -- (2,4) -- (1,4) {};
	\draw[-] (0,0) -- (2,0) -- (2,2) -- (0,2) -- (0,0) {};
	\draw[-] (0,1) -- (2,1) {};
	\draw[-] (1,0) -- (1,2) {};
	\draw[-] (0,4) -- (2,4) -- (2,6) -- (0,6) -- (0,4) {};
	\draw[-] (0,5) -- (2,5) {};
	\draw[-] (1,4) -- (1,6) {};
	
	\draw[->] (0.5,0.5) .. controls (0,1) and (0,4) .. (0.5,4.5) {};
	\draw[->] (0.5,1.5) .. controls (1,2) and (1,5) .. (0.5,5.5) {};
	\draw[->] (1.5,0.5) .. controls (1,1) and (1,4) .. (1.5,4.5) {};
	\draw[->] (1.5,1.5) .. controls (2,2) and (2,5) .. (1.5,5.5) {};
\end{tikzpicture}
}

\subfloat[A translation combined with a permutation]{
\begin{tikzpicture}[x={(0cm,1cm)},y={(1cm,0cm)}]
	\draw[-,fill,opacity=0.2] (0,0) -- (0,1) -- (1,1) -- (1,0) -- (0,0) {};
	\draw[-,fill,opacity=0.1] (0,1) -- (0,2) -- (1,2) -- (1,1) -- (0,1) {};
	\draw[-,fill,opacity=0.3] (1,1) -- (1,2) -- (2,2) -- (2,1) -- (1,1) {};
	\draw[-,fill,opacity=0.4] (1,0) -- (1,1) -- (2,1) -- (2,0) -- (1,0) {};
	\draw[-,fill,opacity=0.2] (0,4) -- (0,5) -- (1,5) -- (1,4) -- (0,4) {};
	\draw[-,fill,opacity=0.4] (0,5) -- (0,6) -- (1,6) -- (1,5) -- (0,5) {};
	\draw[-,fill,opacity=0.3] (1,5) -- (1,6) -- (2,6) -- (2,5) -- (1,5) {};
	\draw[-,fill,opacity=0.1] (1,4) -- (1,5) -- (2,5) -- (2,4) -- (1,4) {};
	\draw[-] (0,0) -- (2,0) -- (2,2) -- (0,2) -- (0,0) {};
	\draw[-] (0,1) -- (2,1) {};
	\draw[-] (1,0) -- (1,2) {};
	\draw[-] (0,4) -- (2,4) -- (2,6) -- (0,6) -- (0,4) {};
	\draw[-] (0,5) -- (2,5) {};
	\draw[-] (1,4) -- (1,6) {};

	\draw[->] (0.5,0.5) .. controls (0,1) and (0,4) .. (0.5,4.5) {};
	\draw[->] (0.5,1.5) -- (1.5,4.5) {};
	\draw[->] (1.5,0.5) -- (0.5,5.5) {};
	\draw[->] (1.5,1.5) .. controls (2,2) and (2,5) .. (1.5,5.5) {};
\end{tikzpicture}
}
\caption{How permutations interact with splittings.}\label{illustratepermutations}
\end{figure}

%

\subsection{Non-deterministic and Probabilistic Computation}

All the preceding results have non-deterministic and probabilistic analogues; we consider in this section the model of non-deterministic graphings. To obtain the same types of results in that case, two issues should be dealt with. First one needs to consider programs of a different type since the result of a non-deterministic computation yield sets of booleans and not a single boolean. Thus, programs will be considered as elements of a more general type than in the deterministic case. We consider here elements of the type $\oc\ListTypeBin\multimap\NBool$, where $\NBool$ is a specific type definable in the models, somehow a non-deterministic version of the booleans.

The second issue concerns acceptance. While it is natural in the deterministic case to ask whether the computation yielded the element $\ttrm{true}\in\Bool$, this is no longer the case. Should one define acceptance as producing at least one element $\ttrm{true}$ or as producing no element $\ttrm{false}$? Both conditions should be considered. In order to obtain a quite general notion of acceptance that can not only capture both notions explained above but extend to other computational paradigms such as probabilistic computation, we use the structure of the realisability models we are working with to define a notion of \emph{test}. Indeed, the models are constructed around an orthogonality relation $\poll$: a test will be an element (or more generally a family of elements) $\testfont{T}$ of the model and a program $P$ accepts the word $\word{w}$ if the execution $P\plug \oc L_{\word{w}}$ is orthogonal to $\testfont{T}$.

Given $P$ an element of $\oc\ListTypeBin\multimap\NBool$ and a test $\testfont{T}$, one can define the language $[P]_{\testfont{T}}$ as the set of all words $\word{w}$ that are accepted by $P$ w.r.t. the test $\testfont{T}$:
$$[P]_{\testfont{T}}=\{\word{w}~|~P\plug \oc L_{\word{w}}\poll \testfont{T}\}$$

\begin{remark}
Acceptance for the specific case of the deterministic model can be defined in this way by considering the right notion of test. 
\end{remark}

\begin{definition}
Let $\Omega$ be a monoid, $\microcosm{m}$ a microcosm, $\testfont{T}$ a test and $\mathcal{L}$ a set of languages. We say the model $\genmodel[]{m}{\Omega}{\ast}$ \emph{characterises the set $\mathcal{L}$ w.r.t. the test $\testfont{T}$} if the set $\{[P]_{\testfont{T}}~|~P\in \oc\ListTypeBin\multimap\NBool\}$ is equal to $\mathcal{L}$.
\end{definition}

In particular, one can show the existence of two tests $\testnl$ and $\testconl$ that correspond to the two notions of acceptance mentioned above and which allows for the characterisation of usual non-deterministic classes. Those are stated in the following theorem. 

\begin{theorem}\label{ndth1}
The model $\nmodel{\{1\}}{m_{\textnormal{i}}}$~$(i\in\naturalN\cup\{\infty\})$ characterises the class \cctwnfa{i} w.r.t. the test $\testnl$ and characterises the class \cctwconfa{i} w.r.t. the test $\testconl$.
\end{theorem}

In the case of probabilistic graphings, one can show the existence of a test $\testprob$ which allows for the characterisation of probabilistic computation with unbounded error. This leads to the following theorems.

\begin{theorem}\label{pth1}
The model $\pmodel{[0,1]}{m_{\textnormal{i}}}$~$(i\in\naturalN\cup\{\infty\})$ characterises the class \cctwpfa{i} w.r.t. the test $\testprob$.
\end{theorem}

\noindent As corollaries of these results, we obtain new characterisations of:
\begin{itemize}[noitemsep,nolistsep]
\item the class \Regular of Regular languages (the model $\nmodel{\{1\}}{m_{\textnormal{1}}}$);
\item the classes \NLogspace and \coNLogspace (the model $\nmodel{\{1\}}{m_{\infty}}$);
\item the class \Stochastic of Stochastic languages (the model $\pmodel{[0,1]}{m_{\textnormal{1}}}$);
\item the class \PLogspace (the model $\pmodel{[0,1]}{m_{\infty}}$).
\end{itemize}

\subsection{Variations}

All of these results are based upon a correspondence between elements of the type $\oc\ListTypeBin\multimap \NBool$ and some kinds of two-way multihead automata, either deterministic, non-deterministic or probabilistic. Several other results can be obtained by modifying the notion of automata considered. We discuss here three possible modifications. 

The first modification is actually a restriction, that is: can we represent computation by \emph{one-way automata}? One can already answer positively to this question, as the two-way capability of the automata does not really find its source in the programs $P$ in $\oc\ListTypeBin\multimap \NBool$ but in the representation of words. One can define an alternative representation of words over an alphabet $\Sigma$ and a corresponding type $\ListTypeOW$. 
We then obtain one-way analogues of \autoref{th1}, \autoref{ndth1}, and \autoref{pth1}.

\begin{theorem}
The model $\dmodel{[0,1]}{m_{\textnormal{i}}}$~$(i\in\naturalN\cup\{\infty\})$ characterises the class \ccowdfa{i}.
\end{theorem}

\begin{theorem}
The model $\nmodel{[0,1]}{m_{\textnormal{i}}}$~$(i\in\naturalN\cup\{\infty\})$ characterises the class \ccownfa{i} w.r.t. the test $\testnl$ and characterises the class \ccowconfa{i} w.r.t. the test $\testconl$.
\end{theorem}

\begin{theorem}
The model $\pmodel{[0,1]}{m_{\textnormal{i}}}$~$(i\in\naturalN\cup\{\infty\})$ characterises the class \ccowpfa{i} w.r.t. the test $\testprob$.
\end{theorem}

The second modification is the extension of automata with a \emph{pushdown stack}. Work in this direction has recently lead to a characterisation of \Ptime in a more syntactical setting \cite{lics-ptime}. Even though the syntactical approach just mentioned could very well be transported to our setting (it was shown \cite{seiller-goig} that elements of the \emph{resolution algebra} can be represented as graphings), this would lead to a characterisation based on a microcosm containing non-measure-preserving maps. Even though non-measure-preserving maps are allowed in the general setting of graphings \cite{seiller-goig}, the use of measure-preserving microcosm is more interesting in view of the possible use of mathematical invariants such as $\ell^{2}$-Betti numbers discussed in the next section. One can find such a measure-preserving microcosm $\mathfrak{p}$ which leads to a characterisation of \Ptime in both the model of deterministic graphings and the model of non-deterministic graphings because non-determinism do not add any expressivity to the model of deterministic two-way multihead automata with a pushdown stack.

The last modification is the consideration of \emph{quantum graphings}, i.e. graphings computing with complex numbers. This is still a work in progress, but we believe that one can define variants of quantum graphings corresponding to measure-once or measure-many quantum automata, leading to several other characterisations.

\section{Complexity Constraints and Graphing Algebras}

It is our belief also that important contributions can be made by using well-established mathematical techniques, tools and invariants from operator algebras and dynamical systems for addressing open problems in complexity. We now show our approach relates to well-known mathematical notions and how this enables new techniques through the computation of invariants such as $\ell^{2}$-Betti numbers. 

\subsection{Compilations}


The results exposed in the previous section lead us to conjecture a correspondence between complexity classes and microcosms, a conjecture that we are now able to state precisely. 

It may appear that the microcosms correspond to complexity \emph{constraints}: for instance $\microcosm{m}_{\infty}$ captures the constraint of \enquote{logarithmic space computation} uniformly, i.e. the same microcosm corresponds to the same constraint for several computational paradigms (sequential, probabilistic). This correspondence, if it holds, can only be partial since, as we already explained, an extension of the microcosm $\microcosm{m}_{\infty}$ which characterise the class \Ptime of polynomial time predicates do not characterise \NPtime in the non-deterministic model. As a consequence, we know that all microcosms do not characterise uniformly a complexity constraint. 

We will not discuss this complex question of whether complexity \emph{constraints} can be described uniformly by a microcosm. The mere question of deciding \emph{what is} a complexity constraint seems as difficult to answer as the question of deciding what is the right notion of algorithm. We therefore consider the already difficult question of a possible (partial) correspondence between microcosms and hierarchies of complexity classes. To formalise this, let us fix once and for all a monoid $\Omega$, a type of graphings -- e.g. probabilistic -- and a test $\testfont{T}$. In the following, we will denote by $\mathtt{C}[\microcosm{m}]$ the set of languages characterised by the type $\ListType\Rightarrow\NBool$ in the chosen model $\anymodel{\Omega}{m}$ w.r.t. the test $\testfont{T}$. 



We now consider the following natural notion of \emph{compilation}. 

\begin{definition}
A measurable map $m:\measured{X}\rightarrow\measured{X}$ is \emph{compilable} in a set of measurable maps $N$ when there exists a finite partition $X_{1},\dots,X_{k}$ of $\measured{X}$ and elements $n_{1},\dots, n_{k}\in N$ such that $m\restr{X_{i}}\aeq (n_{i})\restr{X_{i}}$.
\end{definition}

\begin{notation} 
We denoted by $m\prec_{\textnormal{c}} N$ the fact that $m$ is compilable in $N$. This naturally extends to a preorder on sets of measurable maps that we abusively denote by $\prec_{\textnormal{c}}$ and defined as $M\prec_{\textnormal{c}} N$ if and only if $m\prec_{\textnormal{c}} N$ for all $m\in M$.
\end{notation}

\begin{definition}
We define the equivalence relation on microcosms: 
\[\microcosm{m}\sim_{\textnormal{c}}\microcosm{n} \Leftrightarrow (\microcosm{m}\prec_{\textnormal{c}}\microcosm{n}\wedge \microcosm{n}\prec_{\textnormal{c}}\microcosm{m})\]
\end{definition}


One can then easily show the following result, stating that -- everything else being equal -- two equivalent microcosms give rise to the same complexity classes.

\begin{theorem}\label{thm}
If $\microcosm{m}\prec_{\textnormal{c}}\microcosm{n}$, then $\mathtt{C}[\microcosm{m}]\subset\mathtt{C}[\microcosm{n}]$.
\end{theorem}

\begin{proof}[Proof (sketch).]
Given a graphing $G$ in $\microcosm{m}$, we show how one can define a graphing $\bar{G}$ in $\microcosm{n}$ such that for all representation of a word $W\in\ListTypeBin$, $G\plug W\poll \testfont{T}$ if and only if $\bar{G}\plug W\poll \testfont{T}$. The new graphing $\bar{G}$ is defined by inductively replacing edges in $G$ by a finite family of edges. The reason why the resulting graphing $\bar{G}$ satisfy the wanted property is comes from the fact paths between $G$ and $W$ are compilable in the set of paths between $\bar{G}$ and $W$. Consequently, cycles between $\bar{G}\plug W$ (the paths) and $ \testfont{T}$ used to define the orthogonality turn out to be refinements of cycles between $G\plug W$ and $\testfont{T}$, which implies that $G\plug W\poll \testfont{T}$ if and only if $\bar{G}\plug W\poll \testfont{T}$.
\end{proof}

We conjecture that the converse of Theorem \ref{thm} holds, i.e. if $\microcosm{m}\not\sim_{\textnormal{c}}\microcosm{n}$ are not equivalent, then $\mathtt{C}[\microcosm{m}]\neq\mathtt{C}[\microcosm{n}]$ are distinct.

\begin{conjecture}\label{conjecture}
If $\microcosm{m}\not\sim_{\textnormal{c}}\microcosm{n}$, then $\mathtt{C}[\microcosm{m}]\neq\mathtt{C}[\microcosm{n}]$.
\end{conjecture}

This would provide a partial equivalence between the problem of classifying complexity classes and that of classifying microcosms. As we will see, the first results we obtained are coherent with this conjecture. We explain below how the notion of equivalence $\sim_{\textnormal{c}}$ relates (in specific cases) to well-studied notions in ergodic theory and von Neumann algebras. This fact could allow, in case the conjecture is proven, for the use of mathematical invariants as a proof method for obtaining separation results. 

\subsection{Measurable Equivalence Relations}

Let us now explain how the microcosms, used to characterise complexity classes in the work described above, can (in some cases) generate a \emph{measured equivalence relation}. We first recall the basics about this notion.

The notion of graphing considered in our work is in fact a generalisation of the notion considered in ergodic theory and operator algebras. The usual definition of graphings do not comprise weights. Moreover, those are usually considered built from measure-preserving maps only. We will refer to these weightless graphings as \emph{simple graphings}. 

\begin{definition}
A \emph{weightless graphing} $\phi$ is a $\{1\}$-weighted graphing in the macrocosm $\microcosm{M}$ over $\measured{X}$.

A \emph{simple graphing} $\phi$ is a $\{1\}$-weighted graphing in the microcosm $\microcosm{mp}$ of measure-preserving Borel isomorphisms over $\measured{X}$.
\end{definition}

\begin{remark}
This definition of simple graphings is equivalent to the usual definition \cite{gaboriaucost} as a family of partial measure-preserving Borel isomorphisms $\phi_{f}:S^{\phi}_{f}\rightarrow T^{\phi}_{f}$ where $S^{\phi}_{f}, T^{\phi}_{f}$ are Borel subsets of $\measured{X}$.
\end{remark}

Given a simple graphing, one can define an equivalence relation $\mathcal{R}(\phi)$ by considering the transitive closure of the relation 
\[ \{(x,\phi_{f}(x))~|~ f\in E^{\phi}, x\in S^{\phi}_{f}\} \]
This relation can be described directly in terms of paths. We define for this the set of $\phi$-words as the set:
\[ \words{\phi}=\{f_{1}^{\epsilon_{1}}f_{2}^{\epsilon_{2}}\dots f_{k}^{\epsilon_{k}}~|~ k\geqslant 0, \forall 1\leqslant i\leqslant k, f_{i}\in E^{\phi}, \epsilon_{i}\in\{-1,1\}\}  \]
Each $\phi$-word $\pi=f_{1}^{\epsilon_{1}}f_{2}^{\epsilon_{2}}\dots f_{k}^{\epsilon_{k}}$ defines a partial measure-preserving Borel isomorphism $\phi_{\pi}$ defined as the composite $\phi_{f_{k}}^{\epsilon_{k}}\circ \phi_{f_{k-1}}^{\epsilon_{k-1}}\circ \dots\circ \phi_{f_{1}}^{\epsilon_{1}}$ from its domain $S^{\phi}_{\pi}$ to its codomain $T^{\phi}_{\pi}$.

\begin{definition}
Let $\phi$ be a simple graphing. We define the equivalence relation $\mathcal{R}(\phi)$ as
\[ \mathcal{R}(\phi)=\{(x,y)~|~\exists \pi\in\words{\phi}, x\in S^{\phi}_{\pi}, y=\phi_{pi}(x) \]
\end{definition}

One can then show that this is a \emph{measurable equivalence relation}, or Borel equivalence relation, i.e. an equivalence relation on $X$ such that:
\begin{enumerate}[noitemsep,nolistsep]
\item \label{Boreleq1} $\mathcal{R}(\phi)$ is a Borel subset of $\measured{X}\times\measured{X}$;
\item \label{Boreleq2} the set of equivalence classes of $\mathcal{R}(\phi)$ is countable;
\item \label{Boreleq3} every partial isomorphism $f:A\rightarrow B$ whose graph $\{(a,f(a))~|~a\in A\}$ is a subset of $\mathcal{R}(\phi)$ is measure-preserving.
\end{enumerate}

\begin{remark}
A microcosm $\microcosm{m}$ is a weightless graphing. Thus every microcosm $\microcosm{m}$ included in the microcosm $\microcosm{mp}$ of measure-preserving Borel automorphisms gives rise to a Borel equivalence relation $\mathcal{R}(\microcosm{m})$.
\end{remark}

\begin{proposition}\label{eqboreleq}
Let $\microcosm{m}$ and $\microcosm{n}$ be equivalent microcosms. Then $\mathcal{R}(\microcosm{m})=\mathcal{R}(\microcosm{n})$.
\end{proposition}

Borel equivalence relations were extensively studied in several domains of mathematics, as such equivalence relations are induced by measurable group actions. Indeed, if $\alpha$ is a measure-preserving action of a (discrete countable) group $G$ on a (atom-free) standard Borel space $\measured{X}$, then one can define the Borel equivalence relation:
\[ \mathcal{R}(\alpha)=\{ (x,y)~|~ \exists g\in G, y=\alpha(g)(x) \} \]

Conversely, it was shown by Feldman and Moore \cite{FeldmanMoore1} that, given a measurable equivalence relation $\mathcal{R}$ on a standard Borel space $\measured{X}$, there exists a group $G$ of measure-preserving Borel automorphisms of $\measured{X}$ such that\footnote{In that specific case, we write $G$ the action defined as the inclusion $G\subset \textnormal{Aut}(\measured{X})$.} $\mathcal{R}=\mathcal{R}(G)$.

Trying to classify group actions, mathematicians have developed fine invariants to classify Borel equivalence relations. In particular, Gaboriau extended the notion of $\ell^{2}$-Betti numbers\footnote{A definition which is coherent with $\ell^{2}$-Betti number defined by Atiyah \cite{atiyahl2betti}, the later generalization to groups by Cheeger and Gromov \cite{gromov-l2} reformulated by Lück \cite{luckl2invariants}, and the generalization to von Neumann algebras by Connes and Shlyakhtenko \cite{connesl2}.}  \cite{gaboriaul2} to this setting, and also introduced the notion of \emph{cost} \cite{gaboriaucost} which can be understood as an approximation of the first $\ell^{2}$-Betti number. We here recall the definition of the latter.

\begin{definition}
The cost of a graphing $\phi$ is defined as the (possibly infinite) real number:
\[ \mathcal{C}(\phi)=\sum_{f\in E^{\phi}} \lambda(S^{\phi}_{f}) \]
The cost of a measurable equivalence relation $\mathcal{R}$ is then the infimum of the cost of simple graphings generating $\mathcal{R}$:
\[ \mathcal{C}(\mathcal{R})=\inf \{\mathcal{C}(\phi)~|~\phi\text{ simple graphing s.t. }\mathcal{R}=\mathcal{R}_{\phi}\} \]
The cost of a group $\Gamma$ is the infimum of the cost of the Borel equivalence relations $\mathcal{R}(\alpha)$ defined by free actions of $\Gamma$ on a standard probability space:
\[ \mathcal{C}(\Gamma)=\inf \{\mathcal{C}(\mathcal{R}(\alpha))~|~\alpha\text{ free action of $\Gamma$ onto $\measured{X}$ with $\mu(X)=1$}\} \]
\end{definition}

%

Let us refer to Gaboriau's work for a complete overview \cite{gaboriaucost}, and state the result that will be of use here. A result due to Levitt \cite{levitt_graphings} allows to compute the cost of finite groups, which are shown to be of fixed cost -- i.e. all free actions of finite groups have the same cost. Gaboriau also computed the (fixed) cost of infinite amenable groups. The following proposition states these two results simultaneously.

\begin{proposition}[Gaboriau \cite{gaboriaucost}]
If $\Gamma$ is an amenable group, then $\Gamma$ is has fixed cost $1-\frac{1}{\card{\Gamma}}$, where by convention $\frac{1}{\infty}=0$.
\end{proposition}

\subsection{Measurable Preorders, Homotopy Equivalence}

First, let us have a look at the measurable equivalence relations $\mathcal{R}(\microcosm{m}_{i})$ for $\microcosm{m}_{i}$ the microcosms defined in the previous section. Notice that $\mathcal{R}(\microcosm{m}_{i})=\integerN^{2}\times\mathcal{R}(\microcosm{s}_{i})$ where $\mathcal{R}(\microcosm{s}_{i})$ is the Borel equivalence relation on $[0,1]^{\naturalN}$ generated by the natural group action of $\mathfrak{G}_{i}$ -- the group of permutations over $i$ elements -- on the Hilbert cube $[0,1]^{\naturalN}$ (permutations act on the first $i$ copies of $[0,1]$). This action being free, the Hilbert cube being a standard probability space, and the group $\mathfrak{G}_{i}$ being of finite order, we have that $\mathcal{C}(\mathcal{R}(\microcosm{s}_{i}))=\mathcal{R}(\mathfrak{G}_{i})=1-\frac{1}{i!}$. This shows the following separation theorem for the microcosms $\microcosm{m}_{i}$.


\begin{theorem}
For all $1\leqslant i< j\leqslant \infty$, $\microcosm{m}_{i}\not\sim_{c}\microcosm{m}_{j}$.
\end{theorem}


These results are coherent with \autoref{conjecture}. Indeed, a number of well-known separation results such as \cctwdfa{k}$\neq$\cctwdfa{k+1} \cite{monien} hold for the various notion of automata considered in the previous section.

This result illustrates how one could use invariants to show that two complexity classes are distinct, under the hypothesis that \autoref{conjecture} is true. This uses the fact that invariants such as $\ell^{2}$-Betti numbers can be used to show that two microcosms are not equivalent. This raises two natural questions that we answer now. 

\paragraph{Borel equivalence relations are not enough} First the use of Borel equivalence relations is too restrictive for our purpose. Indeed, although the microcosms $\microcosm{m}_{i}$ were simple graphings, we do not want to restrict the framework to those. Indeed, when one wants to represent computational principles, one sometimes want some non-invertible principles to be available. As a consequence, we are interested in microcosms that are not subsets of the microcosm $\microcosm{mp}$, and which are not groups. In other words, we are not interested in group actions on a space, but rather on \emph{monoid actions}. The problem of classifying monoid actions is however much more complex, and much less studied.

In order to have a finer analysis of the equivalence, we want to distinguish between monoids of measurable maps and groups of measurable maps. Given a weightless graphing $\phi$, we consider the set of \emph{positive $\phi$-words} as the following subset of $\phi$-words.
\[ \wordspos{\phi}=\{f_{1}f_{2}\dots f_{k}~|~ k\geqslant 1, \forall 1\leqslant i\leqslant k, f_{i}\in E^{\phi}\}  \]

\begin{definition}
Given a weightless graphing, we define the $\preorder[\phi]$ as
\[ \preorder[\phi]=\{(x,y)~|~\exists \pi\in\wordspos{\phi}, y=\phi_{\pi}(x)\} \]
\end{definition}

We obtain the following refinement of \autoref{eqboreleq}.

\begin{proposition}
If two microcosms $\microcosm{m}$ and $\microcosm{n}$ are equivalent, the preorders $\preorder[m]$ and $\preorder[n]$ are equal.
\end{proposition}

Can one define invariants such as $\ell^{2}$-Betti numbers in this more general case? The notion of cost can be obviously defined for measurable preorders, but is this still an interesting invariant? In case of $\ell^{2}$-Betti numbers, this question is quite complex as it involves the definition of a von Neumann algebra generated by a measurable preorder. Although some work defining von Neumann algebras from left-cancellable monoids could lead to a partial answer, this is a difficult open question.

\paragraph{Justifying the use of homotopy invariants.}

The second issue is that the invariant mentioned above, namely $\ell^{2}$-Betti numbers, are \emph{homotopy invariants}, i.e. they are used to show that two measurable equivalence relations are \emph{not homotopy equivalent}, which is stronger than showing that they are not equal. So, are those invariants too complex? Couldn't we use easier invariants? The answer to this question lies in a more involved study of the equivalence of microcosms. The notion of compilation we discussed above is not the finest equivalence one could consider. We chose to work with this simplified notion since the most general setting needs to go into detailed discussions about the definitions of the exponential connectives in the model. However, one should actually consider a more involved notion, namely that of \emph{compilation up to Borel automorphisms}. This notion is not complicated to define. 

\begin{definition}
Let $\Theta$ be a set of Borel automorphisms of $\measured{X}$. A measurable map $m$ is compilable in a set $N$ of measurable maps \emph{up to $\Theta$} if and only if there exists $\theta\in\Theta$, a finite partition $X_{1},\dots,X_{k}$ of $\measured{X}$ and elements $n_{1},\dots, n_{k}\in N$ such that $\theta\circ m\restr{X_{i}}\aeq (n_{i})\restr{X_{i}}$.
\end{definition}

Then the corresponding equivalence of two microcosms $\microcosm{m}$ and $\microcosm{n}$ up to a set of Borel automorphisms do not induce the equality of the measurable equivalence relations $\mathcal{R}(\microcosm{m})$ and $\mathcal{R}(\microcosm{n})$, but only that those are homotopy equivalent. In this case, one can understand the interest of invariants such as $\ell^{2}$-Betti numbers, or the cost which can be understood as an approximation of the first $\ell^{2}$-Betti number.

\section{Perspectives}


We believe the above conjecture provides a very good working hypothesis, even if no proof of it is to be found, as the approach we propose provides an homogeneous approach to computational complexity.

This \emph{complexity-through-realisability} theory we sketched is founded on alternative definitions of the notions of algorithms and complexity classes. The techniques were illustrated by first results showing how a large family of (predicate) complexity classes can be characterised by these techniques. Those are a first step towards a demonstration that these definitions capture and generalise standard ones, offering a unified homogeneous framework for the study of complexity classes. 

Future work in this direction should therefore aim at fulfilling two main objectives. The first objective is to establish that this new approach to complexity captures, generalises and extends the techniques developed by previous logic-based approaches to computational complexity. The second objective is to establish that invariants and tools available for the mathematical theories underlying our approach can be used to address open problems in complexity, as already explained in the previous section.


\subsection{A Uniform Approach to Computational Complexity}

We propose the following three goals to deal with the first objective: show the \emph{complexity-through-realisability} techniques (i) are coherent with classical theory, (ii) generalise and improve state-of-the-art techniques, and finally (iii) provide the first homogeneous theory of complexity for several computational paradigms. 

\paragraph{Automata, Turing machines, etc.} The results presented in this paper are a first step toward showing that our technique are coherent with the classical complexity theory. However, all our results lean on previously known characterisations of complexity classes (predicates) by means of different kinds of automata. As explained above, extensions of these notions of automata can be considered to obtain further results. However, it is important to realise that our approach can also deal with other models of computation. 

An adaptation of work by Asperti and Roversi \cite{aspertiroversi} and Baillot \cite{baillot} should allow for encoding Turing machines in some of the realisability models we consider. This should lead to characterisations of several other complexity classes (predicates), such as the exponential hierarchy. In particular this should allow for a characterisation of \NPtime, a class that -- as explained above -- is not characterised naturally by pushdown automata. Moreover, previous characterisations of \NCone and \Pspace by means,  respectively, of \emph{branching programs} (Barrington's theorem \cite{barrington}) and \emph{bottleneck Turing machines} \cite{bottleneck} should lead to characterisations of those classes. 

It would also be interesting to understand if the definition of algorithms as Abstract State Machines (ASMs) proposed by Gurevich \cite{gurevichasm} corresponds to a specific case of our definition of algorithms as graphings. Although no previous work attempted to relate ASMs and GoI, an ASM is intuitively a kind of automata on first-order structures and such objects can be seen as graphings \cite{seiller-goig}. This expected result will show that the notion of graphing provides an adequate mathematical definition of the notion of computation.

\paragraph{Predicates, functions, etc.} All results presented above are characterisations of predicate complexity classes. It is natural to wonder if the approach is limited to those or if it applies similarly to the study of functions complexity classes. 

Since one is considering models of linear logic, the approach does naturally extend to functions complexity classes. A natural lead for studying function classes is to understand the types $\ListType\Rightarrow \ListType$ -- functions from (binary) natural numbers to (binary) natural numbers -- in the models considered. A first step was obtained by the definition of a model of Elementary Linear Logic (\ELL) \cite{seiller-goie}, a system which is known to characterise elementary time functions as the set of algorithms/proofs of type $\oc\ListType\Rightarrow \ListType$. Although the mentioned model is not shown complete for \ELL{}, it provides a first  step in this direction as the type of functions in this model is naturally sound for elementary time functions. We believe that \emph{complexity-through-realisability} techniques extend ICC techniques in that usual complexity classes of functions can be characterised as types of functions $\oc\ListType\Rightarrow \ListType$ in different models. A first natural question is that of the functions computed by elements of this type in the models described in the previous section.

Furthermore, we expect characterisations and/or definitions of complexity classes of higher-order functionals. Indeed, the models considered  contain types of higher-order functionals, e.g. $\cond{(\ListType\Rightarrow \ListType)\Rightarrow(\ListType\Rightarrow \ListType)}$ for type 2 functionals. These models therefore naturally characterise subclasses of higher-order functionals. As no established theory of complexity for higher-order functional exists at the time, this line of research is of particular interest.

\paragraph{Deterministic, Probabilistic, etc.}  As exposed in the previous section, the techniques apply in an homogeneous way to deterministic, non-deterministic and probabilistic automata. Moreover a generalisation towards quantum computation seems natural, as already discussed. The framework offered by graphings is indeed particularly fit for both the quantum and the probabilistic computational paradigm as it is related to operator algebras techniques, hence to both measure theory and linear algebra. In this particular case of probabilistic and quantum computation, we believe that it will allow for characterisations of other standard probabilistic and quantum complexity classes such as \PPtime, \BPPtime, or \BQPtime.

Furthermore, the generality of the approach and the flexibility of the notion of graphing let us hope for application to other computational paradigms. Among those, we can cite concurrent computation and cellular automata. On one hand, this would provide a viable and well-grounded foundation for a theory of computational complexity for concurrent computation, something currently lacking. On the other hand, applying the techniques to cellular automata will raise the question of the relationship between \enquote{classical} complexity theory and the notion of \emph{communication complexity}.

\subsection{Using Mathematical Invariants}

As we already discussed in the previous section, the proposed framework has close ties with some well-studied notions from other fields of mathematics. Of course, the possibility of using invariants from these fields in order to obtain new separation results would be a major step for complexity theory. However, let us notice that even though nothing forbids such results, the proposed method do not provide a miraculous solution to long-standing open problems. Indeed, in order to obtain separation results using the techniques mentioned in the previous sections, one would have to
\begin{itemize}[noitemsep,nolistsep]
\item first prove \autoref{conjecture}, which is not the simplest thing to do;
\item then characterise the two classes to be separated by models $\anymodel{\Omega}{m}$ and  $\anymodel{\Omega}{m}$ which differ only by the chosen microcosm (either $\microcosm{m}$ or $\microcosm{n}$);
\item then compute invariants showing that these microcosms are not homotopy equivalent (if the microcosms are groups of measure-preserving Borel isomorphisms), or -- in the general case -- define and then compute homotopy invariants separating the two microcosms.
\end{itemize}

As an illustration, let us consider what could be done if someone were to find a proof of \autoref{conjecture} today. Based on the results exposed above, one could only hope for already known separation results, namely the fact that the inclusions \cctwdfa{k}$\subsetneq$\cctwdfa{k+1}, \cctwnfa{k}$\subsetneq$\cctwnfa{k+1}, \cctwconfa{k}$\subsetneq$\cctwconfa{k+1} and \cctwpfa{k}$\subsetneq$\cctwpfa{k+1}. Those, as explained above, would be obtained by computing the cost of the associated Borel equivalence relations.

If one were to consider the characterisation of \Ptime obtained syntactically \cite{lics-ptime} as well, one could hope for separating the classes \Logspace and \Ptime{}. However, the microcosm used to characterise \Ptime is neither a group nor consists in measure-preserving maps, two distinct but major handicaps for defining and computing invariants. As explained, future work will show that one can actually consider a microcosm of measure-preserving maps to characterise \Ptime. But once again, this microcosm would not be a group, hence one would need to extends the theory of invariants for Borel equivalence relations in order to hope for a separation result. Last, but not least, these invariants should be expressive enough! For instance, the cost of the relation induced by the microcosm $\microcosm{m}_{\infty}$ is equal to $1$, as is the cost of every amenable infinite group. Amenable extensions of this microcosm will henceforth not be shown to be strictly greater than $\microcosm{m}_{\infty}$ unless some other, finer, invariants are considered and successfully computed!

However, the proposed proof method does not naturally appear as a \emph{natural proof} in the sense of Razborov and Rudich \cite{naturalproofs}. It is the author's opinion that this simple fact on its own makes the approach worth studying further.


\section*{References}
\bibliographystyle{elsarticle-num-names}
\bibliography{thomas}

\begin{thebibliography}{61}
\providecommand{\natexlab}[1]{#1}
\providecommand{\url}[1]{\texttt{#1}}
\providecommand{\urlprefix}{URL }
\expandafter\ifx\csname urlstyle\endcsname\relax
  \providecommand{\doi}[1]{doi:\discretionary{}{}{}#1}\else
  \providecommand{\doi}[1]{doi:\discretionary{}{}{}\begingroup
  \urlstyle{rm}\url{#1}\endgroup}\fi
\providecommand{\bibinfo}[2]{#2}

\bibitem[{Seiller(2014{\natexlab{a}})}]{seiller-goig}
\bibinfo{author}{T.~Seiller}, \bibinfo{title}{Interaction Graphs: Graphings},
  \bibinfo{journal}{Submitted} .

\bibitem[{Aubert and Seiller(2014{\natexlab{a}})}]{seiller-conl}
\bibinfo{author}{C.~Aubert}, \bibinfo{author}{T.~Seiller},
  \bibinfo{title}{Characterizing co-NL by a group action},
  \bibinfo{journal}{Mathematical Structures in Computer Science}
  \doi{\bibinfo{doi}{0.1017/S0960129514000267}}, \bibinfo{note}{to appear}.

\bibitem[{Aubert and Seiller(2014{\natexlab{b}})}]{seiller-lsp}
\bibinfo{author}{C.~Aubert}, \bibinfo{author}{T.~Seiller},
  \bibinfo{title}{Logarithmic Space and Permutations},
  \bibinfo{journal}{Information and Computation} \bibinfo{note}{To appear}.

\bibitem[{Seiller(2015)}]{seiller-goic}
\bibinfo{author}{T.~Seiller}, \bibinfo{title}{Interaction Graphs and Complexity
  {I}}, \bibinfo{note}{in preparation.}, \bibinfo{year}{2015}.

\bibitem[{Hartmanis and Stearns(1965{\natexlab{a}})}]{hartmanisstearns}
\bibinfo{author}{J.~Hartmanis}, \bibinfo{author}{R.~Stearns},
  \bibinfo{title}{On the computational complexity of algorithms},
  \bibinfo{journal}{Transactions of the American Mathematical Society}
  \bibinfo{volume}{117}.

\bibitem[{Immerman(1988)}]{immerman}
\bibinfo{author}{N.~Immerman}, \bibinfo{title}{Nondeterministic space is closed
  under complementation}, in: \bibinfo{booktitle}{Structure in Complexity
  Theory Conference}, \bibinfo{year}{1988}.

\bibitem[{Cobham(1965)}]{cobham}
\bibinfo{author}{A.~Cobham}, \bibinfo{title}{The intrinsic computational
  difficulty of functions}, in: \bibinfo{booktitle}{Proceedings of the 1964
  CLMPS}, \bibinfo{year}{1965}.

\bibitem[{Cook(1971)}]{cookfoundations}
\bibinfo{author}{S.~Cook}, \bibinfo{title}{The complexity of theorem-proving
  procedures}, in: \bibinfo{booktitle}{Proceedings of the 3rd ACM Symposium on
  Theory of Computing}, \bibinfo{year}{1971}.

\bibitem[{Hartmanis and
  Stearns(1965{\natexlab{b}})}]{hartmanisstearnsfoundations}
\bibinfo{author}{J.~Hartmanis}, \bibinfo{author}{R.~Stearns},
  \bibinfo{title}{On the computational complexity of algorithms},
  \bibinfo{journal}{Transactions of the AMS} \bibinfo{volume}{117}.

\bibitem[{Savitch(1970)}]{savitch}
\bibinfo{author}{W.~Savitch}, \bibinfo{title}{Relationship between
  nondeterministic and deterministic tape classes}, \bibinfo{journal}{Journal
  of Computer and Systems Sciences} \bibinfo{volume}{4}.

\bibitem[{Razborov and Rudich(1997)}]{naturalproofs}
\bibinfo{author}{A.~A. Razborov}, \bibinfo{author}{S.~Rudich},
  \bibinfo{title}{Natural proofs}, \bibinfo{journal}{Journal of Computer and
  System Sciences} \bibinfo{volume}{55}.

\bibitem[{Fagin(1974)}]{fagin74}
\bibinfo{author}{R.~Fagin}, \bibinfo{title}{Generalized first-order spectra and
  polynomial-time recognizable sets}, in: \bibinfo{booktitle}{SIAM-AMS Proc.},
  vol.~\bibinfo{volume}{7}, \bibinfo{year}{1974}.

\bibitem[{Jones and Selman(1974)}]{jonesselman}
\bibinfo{author}{N.~Jones}, \bibinfo{author}{A.~Selman}, \bibinfo{title}{Turing
  machines and the spectra of first-order formulas}, \bibinfo{journal}{Journal
  of Symbolic Logic} \bibinfo{volume}{39}.

\bibitem[{Compton and Gr\"{a}del(1996)}]{comptongradel}
\bibinfo{author}{K.~J. Compton}, \bibinfo{author}{E.~Gr\"{a}del},
  \bibinfo{title}{Logical definability of counting functions},
  \bibinfo{journal}{J. Comput. Syst. Sci.} \bibinfo{volume}{53}.

\bibitem[{Dawar and Gr\"{a}del(2008)}]{dawargradel}
\bibinfo{author}{A.~Dawar}, \bibinfo{author}{E.~Gr\"{a}del},
  \bibinfo{title}{The descriptive complexity of parity games},
  \bibinfo{journal}{LMCS} \bibinfo{volume}{5213}.

\bibitem[{Gr\"{a}del and Gurevich(1995)}]{gradelgurevich}
\bibinfo{author}{E.~Gr\"{a}del}, \bibinfo{author}{Y.~Gurevich},
  \bibinfo{title}{Tailoring recursion for complexity},
  \bibinfo{journal}{Journal of Symbolic Logic} \bibinfo{volume}{60}.

\bibitem[{Szelepsc\'{e}nyi(1987)}]{szelepcsenyi}
\bibinfo{author}{R.~Szelepsc\'{e}nyi}, \bibinfo{title}{The method of forced
  enumeration for nondeterministic automata}, \bibinfo{journal}{Bulletin of the
  EATCS} \bibinfo{volume}{33}.

\bibitem[{Bellantoni and Cook(1992)}]{bellantonicook}
\bibinfo{author}{S.~Bellantoni}, \bibinfo{author}{S.~Cook}, \bibinfo{title}{A
  new recursion-theoretic characterization of the polytime functions},
  \bibinfo{journal}{Computational Complexity} \bibinfo{volume}{2}.

\bibitem[{Leivant and Marion(1993)}]{leivantmarion1}
\bibinfo{author}{D.~Leivant}, \bibinfo{author}{J.-Y. Marion},
  \bibinfo{title}{Lambda calculus characterizations of poly-time},
  \bibinfo{journal}{Fundam. Inform.} \bibinfo{volume}{19}.

\bibitem[{Leivant and Marion(1994)}]{leivantmarion2}
\bibinfo{author}{D.~Leivant}, \bibinfo{author}{J.-Y. Marion},
  \bibinfo{title}{Ramified recurrence and computational complexity {II}:
  Substitution and poly-space}, \bibinfo{journal}{Lecture Notes in Computer
  Science} \bibinfo{volume}{933}.

\bibitem[{Girard(1987)}]{ll}
\bibinfo{author}{J.-Y. Girard}, \bibinfo{title}{Linear logic},
  \bibinfo{journal}{Theoretical Computer Science}
  \bibinfo{volume}{50}~(\bibinfo{number}{1}) (\bibinfo{year}{1987})
  \bibinfo{pages}{1--101}, \doi{\bibinfo{doi}{10.1016/0304-3975(87)90045-4}}.

\bibitem[{Girard et~al.(1992)Girard, Scedrov, and Scott}]{BLL}
\bibinfo{author}{J.-Y. Girard}, \bibinfo{author}{A.~Scedrov},
  \bibinfo{author}{P.~J. Scott}, \bibinfo{title}{Bounded linear logic: a
  modular approach to polynomial-time computability}, \bibinfo{journal}{Theor.
  Comput. Sci.} \bibinfo{volume}{97}~(\bibinfo{number}{1})
  (\bibinfo{year}{1992}) \bibinfo{pages}{1--66}, ISSN
  \bibinfo{issn}{0304-3975}, \doi{\bibinfo{doi}{10.1016/0304-3975(92)90386-T}},
  \urlprefix\url{http://dx.doi.org/10.1016/0304-3975(92)90386-T}.

\bibitem[{Danos and Joinet(2003)}]{danosjoinet}
\bibinfo{author}{V.~Danos}, \bibinfo{author}{J.-B. Joinet},
  \bibinfo{title}{Linear Logic \& Elementary Time},
  \bibinfo{journal}{Information and Computation}
  \bibinfo{volume}{183}~(\bibinfo{number}{1}) (\bibinfo{year}{2003})
  \bibinfo{pages}{123--137},
  \doi{\bibinfo{doi}{10.1016/S0890-5401(03)00010-5}}.

\bibitem[{Girard(1989{\natexlab{a}})}]{towards}
\bibinfo{author}{J.-Y. Girard}, \bibinfo{title}{Towards a geometry of
  interaction}, in: \bibinfo{booktitle}{Proceedings of the AMS Conference on
  Categories, Logic and Computer Science}, \bibinfo{year}{1989}{\natexlab{a}}.

\bibitem[{Kleene(1945)}]{kleene}
\bibinfo{author}{S.~C. Kleene}, \bibinfo{title}{On the interpretation of
  intuitionistic number theory}, \bibinfo{journal}{Journal of Symbolic Logic}
  \bibinfo{volume}{10}.

\bibitem[{Krivine(2001)}]{krivine}
\bibinfo{author}{J.-L. Krivine}, \bibinfo{title}{Typed lambda-calculus in
  classical Zermelo-Fraenkel set theory}, \bibinfo{journal}{Arch. Mathematical
  Logic} \bibinfo{volume}{40}.

\bibitem[{Girard(1989{\natexlab{b}})}]{goi1}
\bibinfo{author}{J.-Y. Girard}, \bibinfo{title}{Geometry of Interaction
  $\text{I}$: Interpretation of System {F}}, in: \bibinfo{booktitle}{In Proc.
  Logic Colloquium 88}, \bibinfo{year}{1989}{\natexlab{b}}.

\bibitem[{Gonthier et~al.(1992)Gonthier, Abadi, and
  L{\'e}vy}]{AbadiGonthierLevy92b}
\bibinfo{author}{G.~Gonthier}, \bibinfo{author}{M.~Abadi},
  \bibinfo{author}{J.-J. L{\'e}vy}, \bibinfo{title}{The Geometry of Optimal
  Lambda Reduction}, in: \bibinfo{editor}{R.~Sethi} (Ed.),
  \bibinfo{booktitle}{POPL}, \bibinfo{publisher}{ACM Press}, ISBN
  \bibinfo{isbn}{0-89791-453-8}, \bibinfo{pages}{15--26},
  \doi{\bibinfo{doi}{10.1145/143165.143172}}, \bibinfo{year}{1992}.

\bibitem[{Lamping(1990)}]{Lamping90}
\bibinfo{author}{J.~Lamping}, \bibinfo{title}{An Algorithm for Optimal Lambda
  Calculus Reduction}, in: \bibinfo{editor}{F.~E. Allen} (Ed.),
  \bibinfo{booktitle}{POPL}, \bibinfo{publisher}{ACM Press}, ISBN
  \bibinfo{isbn}{0-89791-343-4}, \bibinfo{pages}{16--30},
  \doi{\bibinfo{doi}{10.1145/96709.96711}}, \bibinfo{year}{1990}.

\bibitem[{Baillot and Pedicini(2001)}]{baillotpedicini}
\bibinfo{author}{P.~Baillot}, \bibinfo{author}{M.~Pedicini},
  \bibinfo{title}{Elementary complexity and geometry of interaction},
  \bibinfo{journal}{Fundamenta Informaticae}
  \bibinfo{volume}{45}~(\bibinfo{number}{1-2}) (\bibinfo{year}{2001})
  \bibinfo{pages}{1--31}.

\bibitem[{Lago(2009)}]{Lago}
\bibinfo{author}{U.~D. Lago}, \bibinfo{title}{The Geometry of Linear
  Higher-order Recursion}, \bibinfo{journal}{ACM Trans. Comput. Logic}
  \bibinfo{volume}{10}~(\bibinfo{number}{2}) (\bibinfo{year}{2009})
  \bibinfo{pages}{8:1--8:38}, ISSN \bibinfo{issn}{1529-3785},
  \doi{\bibinfo{doi}{10.1145/1462179.1462180}},
  \urlprefix\url{http://doi.acm.org/10.1145/1462179.1462180}.

\bibitem[{Girard(2012)}]{normativity}
\bibinfo{author}{J.-Y. Girard}, \bibinfo{title}{Normativity in Logic}, in:
  \bibinfo{editor}{P.~Dybjer}, \bibinfo{editor}{S.~Lindstr{\"o}m},
  \bibinfo{editor}{E.~Palmgren}, \bibinfo{editor}{G.~Sundholm} (Eds.),
  \bibinfo{booktitle}{Epistemology versus Ontology}, vol.~\bibinfo{volume}{27}
  of \emph{\bibinfo{series}{Logic, Epistemology, and the Unity of Science}},
  \bibinfo{publisher}{Springer}, ISBN \bibinfo{isbn}{978-94-007-4434-9,
  978-94-007-4435-6}, \bibinfo{pages}{243--263}, \bibinfo{year}{2012}.

\bibitem[{Aubert et~al.(2014)Aubert, Bagnol, Pistone, and Seiller}]{aplas14}
\bibinfo{author}{C.~Aubert}, \bibinfo{author}{M.~Bagnol},
  \bibinfo{author}{P.~Pistone}, \bibinfo{author}{T.~Seiller},
  \bibinfo{title}{Logic Programming and Logarithmic Space}, in:
  \bibinfo{editor}{J.~Garrigue} (Ed.), \bibinfo{booktitle}{Programming
  Languages and Systems - 12th Asian Symposium, {APLAS} 2014, Singapore,
  November 17-19, 2014, Proceedings}, vol. \bibinfo{volume}{8858} of
  \emph{\bibinfo{series}{Lecture Notes in Computer Science}},
  \bibinfo{publisher}{Springer}, ISBN \bibinfo{isbn}{978-3-319-12735-4},
  \bibinfo{pages}{39--57}, \doi{\bibinfo{doi}{10.1007/978-3-319-12736-1_3}},
  \urlprefix\url{http://dx.doi.org/10.1007/978-3-319-12736-1_3},
  \bibinfo{year}{2014}.

\bibitem[{Aubert et~al.(2015)Aubert, Bagnol, and Seiller}]{lics-ptime}
\bibinfo{author}{C.~Aubert}, \bibinfo{author}{M.~Bagnol},
  \bibinfo{author}{T.~Seiller}, \bibinfo{title}{Memorization for Unary Logic
  Programming: Characterizing Ptime}, \bibinfo{journal}{Submitted} .

\bibitem[{Girard(2006)}]{feedback}
\bibinfo{author}{J.-Y. Girard}, \bibinfo{title}{Geometry of Interaction
  $\text{IV}$: the Feedback Equation}, in:
  \bibinfo{editor}{Stoltenberg-Hansen}, \bibinfo{editor}{Väänänen} (Eds.),
  \bibinfo{booktitle}{Logic Colloquium '03}, \bibinfo{pages}{76--117},
  \bibinfo{year}{2006}.

\bibitem[{Girard(1995)}]{goi3}
\bibinfo{author}{J.-Y. Girard}, \bibinfo{title}{Geometry Of Interaction
  $\text{III}$: Accommodating The Additives}, in: \bibinfo{booktitle}{Advances
  in Linear Logic}, no. \bibinfo{number}{222} in \bibinfo{series}{Lecture Notes
  Series}, \bibinfo{publisher}{Cambridge University Press},
  \bibinfo{pages}{329--389}, \bibinfo{year}{1995}.

\bibitem[{Seiller(2012{\natexlab{a}})}]{seiller-phd}
\bibinfo{author}{T.~Seiller}, \bibinfo{title}{Logique dans le facteur hyperfini
  : g{\'e}ometrie de l'interaction et complexit{\'e}}, Ph.D. thesis,
  \bibinfo{school}{Universit{\'e} Aix-Marseille},
  \urlprefix\url{http://tel.archives-ouvertes.fr/tel-00768403/},
  \bibinfo{year}{2012}{\natexlab{a}}.

\bibitem[{Seiller(2014{\natexlab{b}})}]{seiller-masas}
\bibinfo{author}{T.~Seiller}, \bibinfo{title}{A Correspondence between Maximal
  Abelian Sub-Algebras and Linear Logic Fragments},
  \bibinfo{journal}{Submitted} .

\bibitem[{Sinclair and Smith(2008)}]{FiniteVNAandMasas}
\bibinfo{author}{A.~Sinclair}, \bibinfo{author}{R.~Smith},
  \bibinfo{title}{Finite von Neumann algebras and Masas}, no.
  \bibinfo{number}{351} in \bibinfo{series}{London Mathematical Society Lecture
  Note Series}, \bibinfo{publisher}{Cambridge University Press},
  \bibinfo{year}{2008}.

\bibitem[{Seiller(2012{\natexlab{b}})}]{seiller-goim}
\bibinfo{author}{T.~Seiller}, \bibinfo{title}{Interaction Graphs:
  Multiplicatives}, \bibinfo{journal}{Annals of Pure and Applied Logic}
  \bibinfo{volume}{163} (\bibinfo{year}{2012}{\natexlab{b}})
  \bibinfo{pages}{1808--1837}, \doi{\bibinfo{doi}{10.1016/j.apal.2012.04.005}}.

\bibitem[{Seiller(2014{\natexlab{c}})}]{seiller-goia}
\bibinfo{author}{T.~Seiller}, \bibinfo{title}{Interaction Graphs: Additives},
  \bibinfo{journal}{Accepted for publication in Annals of Pure and Applied
  Logic} .

\bibitem[{Seiller(2014{\natexlab{d}})}]{seiller-goie}
\bibinfo{author}{T.~Seiller}, \bibinfo{title}{Interaction Graphs:
  Exponentials}, \bibinfo{journal}{Submitted} .

\bibitem[{Adams(1990)}]{adams}
\bibinfo{author}{S.~Adams}, \bibinfo{title}{Trees and amenable equivalence
  relations}, \bibinfo{journal}{Ergodic Theory and Dynamical Systems}
  \bibinfo{volume}{10} (\bibinfo{year}{1990}) \bibinfo{pages}{1--14}.

\bibitem[{Levitt(1995)}]{levitt_graphings}
\bibinfo{author}{G.~Levitt}, \bibinfo{title}{On the cost of generating an
  equivalence relation}, \bibinfo{journal}{Ergodic Theory and Dynamical
  Systems} \bibinfo{volume}{15} (\bibinfo{year}{1995})
  \bibinfo{pages}{1173--1181}, \doi{\bibinfo{doi}{10.1017/S0143385700009846}}.

\bibitem[{Gaboriau(2000)}]{gaboriaucost}
\bibinfo{author}{D.~Gaboriau}, \bibinfo{title}{Co{\^u}t des relations
  d'{\'e}quivalence et des groupes}, \bibinfo{journal}{Inventiones
  Mathematicae} \bibinfo{volume}{139} (\bibinfo{year}{2000})
  \bibinfo{pages}{41--98}.

\bibitem[{Feldman and Moore(1977)}]{FeldmanMoore1}
\bibinfo{author}{J.~Feldman}, \bibinfo{author}{C.~C. Moore},
  \bibinfo{title}{Ergodic equivalence relations, cohomology, and von Neumann
  algebras. {I}}, \bibinfo{journal}{Transactions of the American mathematical
  society} \bibinfo{volume}{234}~(\bibinfo{number}{2}) (\bibinfo{year}{1977})
  \bibinfo{pages}{289--324}.

\bibitem[{Fuglede and Kadison(1952)}]{FKdet}
\bibinfo{author}{B.~Fuglede}, \bibinfo{author}{R.~V. Kadison},
  \bibinfo{title}{Determinant theory in finite factors},
  \bibinfo{journal}{Annals of Mathematics}
  \bibinfo{volume}{56}~(\bibinfo{number}{2}).

\bibitem[{Girard(2011)}]{goi5}
\bibinfo{author}{J.-Y. Girard}, \bibinfo{title}{Geometry of Interaction
  $\text{V}$: Logic in the Hyperfinite Factor.}, \bibinfo{journal}{Theoretical
  Computer Science} \bibinfo{volume}{412} (\bibinfo{year}{2011})
  \bibinfo{pages}{1860--1883}.

\bibitem[{Seiller(2014{\natexlab{e}})}]{seiller-lcc14}
\bibinfo{author}{T.~Seiller}, \bibinfo{title}{Interaction Graphs and
  Complexity}, \bibinfo{note}{extended Abstract},
  \bibinfo{year}{2014}{\natexlab{e}}.

\bibitem[{Danos(1990)}]{danos-phd}
\bibinfo{author}{V.~Danos}, \bibinfo{title}{La Logique Lin\'{e}aire
  Appliqu\'{e}e \`a l'\'{E}tude de Divers Processus de Normalisation
  (principalement du $\lambda$-calcul)}, Ph.D. thesis, \bibinfo{school}{Paris
  {VII} University}, \bibinfo{year}{1990}.

\bibitem[{Atiyah(1976)}]{atiyahl2betti}
\bibinfo{author}{M.~Atiyah}, \bibinfo{title}{Elliptic operators, discrete
  groups and von Neumann algebras}, vol. \bibinfo{volume}{32-33} of
  \emph{\bibinfo{series}{Astérisque}}, \bibinfo{publisher}{Société Mathématique
  de France}, \bibinfo{pages}{43--72}, \bibinfo{year}{1976}.

\bibitem[{Cheeger and Gromov(1986)}]{gromov-l2}
\bibinfo{author}{J.~Cheeger}, \bibinfo{author}{M.~Gromov},
  \bibinfo{title}{$L^{2}$-Cohomology and group cohomology},
  \bibinfo{journal}{Topology} \bibinfo{volume}{25}~(\bibinfo{number}{2})
  (\bibinfo{year}{1986}) \bibinfo{pages}{189--215}.

\bibitem[{Lück(2002)}]{luckl2invariants}
\bibinfo{author}{W.~Lück}, \bibinfo{title}{L$^{2}$-Invariants: Theory and
  Applications to Geometry and K-Theory}, vol.~\bibinfo{volume}{44} of
  \emph{\bibinfo{series}{A Series of Modern Surveys in Mathematics}},
  \bibinfo{year}{2002}.

\bibitem[{Connes and Shlyakhtenko(2005)}]{connesl2}
\bibinfo{author}{A.~Connes}, \bibinfo{author}{D.~Shlyakhtenko},
  \bibinfo{title}{L 2-homology for von Neumann algebras},
  \bibinfo{journal}{Journal f{\"u}r die reine und angewandte Mathematik}
  \bibinfo{volume}{2005}~(\bibinfo{number}{586}) (\bibinfo{year}{2005})
  \bibinfo{pages}{125--168}.

\bibitem[{Gaboriau(2002)}]{gaboriaul2}
\bibinfo{author}{D.~Gaboriau}, \bibinfo{title}{Invariants l2 de relations
  d’{\'e}quivalence et de groupes}, \bibinfo{journal}{Publ. Math. Inst. Hautes
  {\'E}tudes Sci} \bibinfo{volume}{95}~(\bibinfo{number}{93-150})
  (\bibinfo{year}{2002}) \bibinfo{pages}{15--28}.

\bibitem[{Monien(1976)}]{monien}
\bibinfo{author}{B.~Monien}, \bibinfo{title}{Transformational Methods and their
  Application to Complexity Problems}, \bibinfo{journal}{Acta Informatica}
  \bibinfo{volume}{6} (\bibinfo{year}{1976}) \bibinfo{pages}{95--108}.

\bibitem[{Asperti and Roversi(2002)}]{aspertiroversi}
\bibinfo{author}{A.~Asperti}, \bibinfo{author}{L.~Roversi},
  \bibinfo{title}{Intuitionistic light affine logic}, \bibinfo{journal}{ACM
  Transactions on Computational Logic (TOCL)}
  \bibinfo{volume}{3}~(\bibinfo{number}{1}) (\bibinfo{year}{2002})
  \bibinfo{pages}{137--175}.

\bibitem[{Baillot(2011)}]{baillot}
\bibinfo{author}{P.~Baillot}, \bibinfo{title}{Elementary Linear Logic Revisited
  for Polynomial Time and an Exponential Time Hierarchy}, in:
  \bibinfo{editor}{H.~Yang} (Ed.), \bibinfo{booktitle}{APLAS}, vol.
  \bibinfo{volume}{7078} of \emph{\bibinfo{series}{Lecture Notes in Computer
  Science}}, \bibinfo{publisher}{Springer}, ISBN
  \bibinfo{isbn}{978-3-642-25317-1}, \bibinfo{pages}{337--352},
  \bibinfo{year}{2011}.

\bibitem[{Barrington(1989)}]{barrington}
\bibinfo{author}{D.~A. Barrington}, \bibinfo{title}{Bounded-width
  polynomial-size branching programs recognize exactly those languages in
  \{NC1\}}, \bibinfo{journal}{Journal of Computer and System Sciences}
  \bibinfo{volume}{38}~(\bibinfo{number}{1}) (\bibinfo{year}{1989})
  \bibinfo{pages}{150 -- 164}.

\bibitem[{Cai and Furst(1991)}]{bottleneck}
\bibinfo{author}{J.-Y. Cai}, \bibinfo{author}{M.~Furst},
  \bibinfo{title}{\Pspace Survives Constant-With Bottlenecks},
  \bibinfo{journal}{International Journal of Foundations of Computer Science}
  \bibinfo{volume}{02}~(\bibinfo{number}{01}) (\bibinfo{year}{1991})
  \bibinfo{pages}{67--76}.

\bibitem[{Gurevich(1995)}]{gurevichasm}
\bibinfo{author}{Y.~Gurevich}, \bibinfo{title}{Specification and Validation
  Methods}, chap. \bibinfo{chapter}{Evolving Algebras 1993: Lipari Guide},
  \bibinfo{publisher}{Oxford University Press, Inc.}, \bibinfo{address}{New
  York, NY, USA}, ISBN \bibinfo{isbn}{0-19-853854-5}, \bibinfo{pages}{9--36},
  \bibinfo{year}{1995}.

\end{thebibliography}

\end{document}